\documentclass[a4paper]{article}
\usepackage{enumitem}
\usepackage{calc}

\usepackage{lmodern} 

\usepackage{graphicx} 

\usepackage{amsmath, accents}
\usepackage{amssymb,tikz}
\usepackage{pict2e}
\usepackage{amsthm}
\usepackage{amscd}
\usepackage{mathrsfs}
\usepackage{todonotes}
\usetikzlibrary{calc} 

\usepackage{yfonts}

\usepackage{marvosym}
\usepackage[percent]{overpic}

\newtheorem{theorem}{Theorem} [section]
	
\newtheorem{corollary}[theorem]{Corollary}	
\newtheorem{lemma}[theorem]{Lemma}

\newtheorem{remark}[theorem]{Remark}
\theoremstyle{definition}

\newcommand{\R}{\mathbb{R}}

\tikzset{
	master/.style={
		execute at end picture={
			\coordinate (lower right) at (current bounding box.south east);
			\coordinate (upper left) at (current bounding box.north west);
		}
	},
	slave/.style={
		execute at end picture={
			\pgfresetboundingbox
			\path (upper left) rectangle (lower right);
		}
	}
}

\let\oldbibliography\thebibliography
\renewcommand{\thebibliography}[1]{\oldbibliography{#1}
\setlength{\itemsep}{-0.5pt}}

\def\XXint#1#2#3{{\setbox0=\hbox{$#1{#2#3}{\int}$}
\vcenter{\hbox{$#2#3$}}\kern-.5\wd0}}

\usepackage{tikz}
\usetikzlibrary{arrows}
\usetikzlibrary{decorations.pathmorphing}
\usetikzlibrary{decorations.markings}
\usetikzlibrary{patterns}
\usetikzlibrary{automata}
\usetikzlibrary{positioning}
\usepackage{tikz-cd}
\tikzset{->-/.style={decoration={
				markings,
				mark=at position #1 with {\arrow{latex}}},postaction={decorate}}}
	
	\tikzset{-<-/.style={decoration={
				markings,
				mark=at position #1 with {\arrowreversed{latex}}},postaction={decorate}}}

\usetikzlibrary{shapes.misc}\tikzset{cross/.style={cross out, draw, 
         minimum size=2*(#1-\pgflinewidth), 
         inner sep=0pt, outer sep=0pt}}

\usepackage{pgfplots}
	
\allowdisplaybreaks	

\numberwithin{equation}{section}

\def\bigO{{\cal O}}

\usepackage[colorlinks=true]{hyperref}
\hypersetup{urlcolor=blue, citecolor=red, linkcolor=blue}

\begin{document}
\title{\vspace*{-1.5cm} Exponential moments for disk counting statistics \\ of random normal matrices in the critical regime}
\author{Christophe Charlier\footnote{Centre for Mathematical Sciences, Lund University, 22100 Lund, Sweden. e-mail: christophe.charlier@math.lu.se
} \, and Jonatan Lenells\footnote{Department of Mathematics, KTH Royal Institute of Technology, 10044 Stockholm, Sweden. e-mail: jlenells@kth.se}}

\maketitle

\begin{abstract}
We obtain large $n$ asymptotics for the $m$-point moment generating function of the disk counting statistics of the Mittag-Leffler ensemble. We focus on the critical regime where all disk boundaries are merging at speed $n^{-\smash{\frac{1}{2}}}$, either in the bulk or at the edge. As corollaries, we obtain two central limit theorems and precise large $n$ asymptotics of all joint cumulants (such as the covariance) of the disk counting function. Our results can also be seen as large $n$ asymptotics for $n\times n$ determinants with merging planar discontinuities.
\end{abstract}


\section{Introduction and statement of results}\label{section: introduction}
In recent years, there has been a growing interest in both the physics and mathematics literature in understanding the counting statistics of various two-dimensional point processes, see e.g. \cite{LeeRiser2016, CE2020, LMS2018, L et al 2019, ES2020, FenzlLambert, SDMS2020, Charlier 2d jumps, SDMS2021, BGL2022, AkemannSungsoo}.
Most of the focus, so far, has been on the one-point counting statistics (see however \cite{FenzlLambert,Charlier 2d jumps}).  In this work we study the $m$-point counting statistics for general $m \in \mathbb{N}_{>0}$ of the Mittag-Leffler ensemble in the critical regime where all disk boundaries are merging. 

\medskip The Mittag-Leffler ensemble with parameters $b>0$ and $\alpha > -1$ is the following probability density function for $n$ points in the complex plane
\begin{align}\label{def of point process}
\frac{1}{n!Z_{n}} \prod_{1 \leq j < k \leq n} |z_{k} -z_{j}|^{2} \prod_{j=1}^{n}|z_{j}|^{2\alpha}e^{-n |z_{j}|^{2b}}, \qquad z_{1},\ldots,z_{n} \in \mathbb{C},
\end{align}
where $Z_{n}$ is the normalization constant. This is a two-dimensional determinantal point process arising in the random normal matrix model \cite{Mehta} and generalizing the complex Ginibre process (which corresponds to $(b,\alpha)=(1,0)$ \cite{Ginibre}). As $n \to + \infty$, the zero counting measure of the average characteristic polynomial of \eqref{def of point process} converges weakly to the measure $\mu(d^{2}z) = \smash{\frac{b^{\smash{2}}}{\pi}}|z|^{\smash{2b-2}}d^{\smash{2}}z$, and the support of $\mu$ is the disk centered at $0$ of radius $b^{-\smash{\frac{1}{2b}}}$ \cite{SaTo}. The Mittag-Leffler ensemble has been widely studied over the years, see e.g. \cite{CZ1998, AV2003, FBKSZ2012, AK2013, AKS2018, BS2021} for various universality and finite-$n$ results. We also refer the interested reader to \cite{HKPV2010} for more background on two-dimensional point processes. 

\newpage For $y>0$, we let $\mathrm{N}(y):=\#\{z_{j}: |z_{j}| < y\}$, i.e. $\mathrm{N}(y)$ is the random variable that counts the number of points in the disk centered at $0$ of radius $y$. In this paper we study the joint statistics of $\mathrm{N}(r_{1}),\ldots,\mathrm{N}(r_{m})$, where $m \in \mathbb{N}_{>0}$ is arbitrary but fixed, in the critical situation where $n \to + \infty$ and all radii are merging near a certain value $r \in (0,b^{-\frac{1}{2b}}]$:
\begin{align}\label{def of rell}
0 < r_{1} < \ldots <r_{m}, \qquad r_{\ell} = r \bigg( 1+\frac{\sqrt{2}\, \mathfrak{s}_{\ell}}{r^{b}\sqrt{n}} \bigg)^{\frac{1}{2b}}, \qquad \mathfrak{s}_{\ell}\in \mathbb{R},
\end{align}
see also Figure \ref{fig: disk counting statistics}. The case $r \in (0,b^{-\frac{1}{2b}})$ corresponds to ``the bulk regime" and $r = b^{-\frac{1}{2b}}$ to ``the edge regime". Our main results can be summarized as follows:
\begin{itemize}
\item Theorems \ref{thm:main thm} and \ref{thm:main thm edge} give precise large $n$ asymptotics for the moment generating function $\mathbb{E}\big[ \prod_{j=1}^{m} e^{u_{j}\mathrm{N}(r_{j})} \big]$, up to and including the fourth term of order $n^{-\frac{\smash{1}}{2}}$. Theorem \ref{thm:main thm} deals with the bulk regime and Theorem \ref{thm:main thm edge} deals with the edge regime. 
\item Corollary \ref{coro:correlation} $(a)$ establishes precise large $n$ asymptotics for all joint cumulants of $\mathrm{N}(r_{1}),\ldots,$ $\mathrm{N}(r_{m})$ in the bulk regime. The analogue of that for the edge regime is given in Corollary \ref{coro:correlation} $(c)$. We also obtain several central limit theorems for the joint fluctuations of $\mathrm{N}(r_{1}),\ldots,\mathrm{N}(r_{m})$; Corollary \ref{coro:correlation} $(b)$ concerns the bulk regime and Corollary \ref{coro:correlation} $(d)$ the edge regime. 
\end{itemize}

\begin{figure}
\begin{center}
\begin{tikzpicture}[master]
\node at (0,0) {\includegraphics[width=5cm]{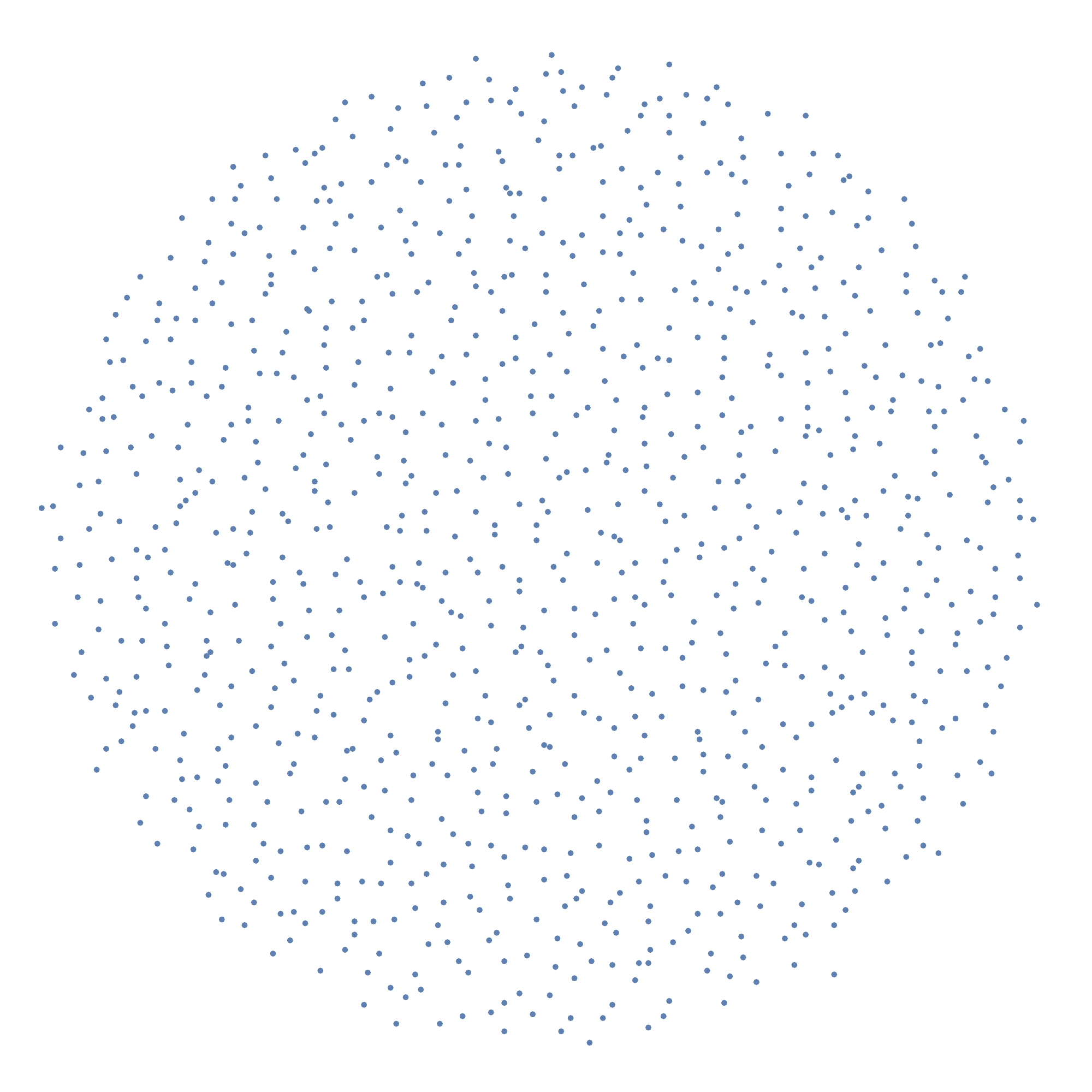}};
\draw[black] (0,0) circle (0.85cm);
\draw[black] (0,0) circle (1cm);
\draw[black] (0,0) circle (1.08cm);
\end{tikzpicture} \hspace{1.5cm} \begin{tikzpicture}[slave]
\node at (0,0) {\includegraphics[width=5cm]{Ginibre_1000.pdf}};
\draw[black] (0,0) circle (2.13cm);
\draw[black] (0,0) circle (2.3cm);
\end{tikzpicture}
\end{center}
\caption{\label{fig: disk counting statistics}Left: three merging disks in the bulk; this case is covered by Theorem \ref{thm:main thm} with $m=3$. Right: two merging disks at the edge; this case is covered by Theorem \ref{thm:main thm edge} with $m=2$. For both pictures, $b=1$ and $\alpha=0$.}
\end{figure}
The problem of determining the large $n$ asymptotics for the one-point generating function, i.e. the case $m=1$ in our setting, was already considered in \cite{CE2020, L et al 2019, FenzlLambert}. The work \cite{CE2020} considers counting statistics of Ginibre-type ensembles in a more general geometric setting, and second order asymptotics for the generating function of counting statistics of general domains (not just centered disks) are obtained in \cite[Proposition 8.1]{CE2020}. The work \cite{L et al 2019} focuses on disk counting statistics of rotation-invariant ensembles with a general potential, and second order asymptotics for the generating function are given in \cite[eq (34)]{L et al 2019}. More precise asymptotics, including the third term, were obtained in \cite[Proposition 2.3]{FenzlLambert} for the Ginibre ensemble (see also \cite[Propositions 2.4 and 2.9]{FenzlLambert} for analogous results on other rotation-invariant processes). Large $n$ asymptotics for $\mathbb{E}\big[ \prod_{j=1}^{m} e^{u_{j}\mathrm{N}(r_{j})} \big]$, including the fourth term and for general $m \in \mathbb{N}_{>0}$, were then obtained in \cite{Charlier 2d jumps} for Mittag-Leffler ensembles, in the case where $r_{1},\ldots,r_{m}$ are fixed (independent of $n$). This case contrasts with the regime \eqref{def of rell} considered here. In the setting of \cite{Charlier 2d jumps}, the general case $m \geq 2$ is actually not much different from the case $m=1$. Indeed, it follows from \cite[Theorem 1.1]{Charlier 2d jumps} that for general $b>0,\alpha>-1$, $u_{1},\ldots,u_{m}\in \mathbb{R}$ and $r_{1}<\ldots<r_{m}$ fixed, we have
\begin{align}\label{decoupling}
\mathbb{E}\bigg[ \prod_{j=1}^{m} e^{u_{j}\mathrm{N}(r_{j})} \bigg] = \prod_{j=1}^{m}\mathbb{E}\Big[ e^{u_{j}\mathrm{N}(r_{j})} \Big] \times \bigg( 1+\bigO\bigg( \frac{(\ln n)^{2}}{n} \bigg) \bigg), \qquad \mbox{as } n \to + \infty.
\end{align}
In fact, by straightforward modifications of the proof of \cite{Charlier 2d jumps}, the error term in \eqref{decoupling} can be shown to be exponentially small, see also \cite[Remark 1.5]{Charlier 2d jumps}. In other words, when $r_{1},\ldots,r_{m}$ are fixed, the $m$-point generating function can be expressed asymptotically as the product of $m$ one-point generating functions (up to an exponentially small error term). This, in turn, implies that all cumulants involving two random variables or more among $\mathrm{N}(r_{1}),\ldots,\mathrm{N}(r_{m})$ are exponentially small --- for the covariance, this fact (namely that $\mbox{Cov}(\mathrm{N}(r_{1}),\mathrm{N}(r_{2}))=\bigO(e^{-cn})$ as $n \to + \infty$) was already noticed in \cite[Theorem 1.7]{Rider} for $(b,\alpha)=(1,0)$, and for the joint fluctuations of $\mathrm{N}(r_{1}),\ldots,\mathrm{N}(r_{m})$, this fact (namely that they become independent Gaussian random variables in the large $n$ limit) was already proved in \cite[Proposition 1.3]{FenzlLambert} for Ginibre-type ensembles. (We mention en passant that the decoupling \eqref{decoupling} is a particular feature of two-dimensional point processes such as \eqref{def of point process}. Indeed, the analogues of \eqref{decoupling} for the one-dimensional sine, Airy, Bessel and Pearcey point processes involve explicit constant pre-factors of order 1, and the associated covariances are not small but of order $1$, see e.g. \cite{ChCl3, ChMor}.) In the regime considered here, namely \eqref{def of rell}, the $m$-point generating function does not decouple as in \eqref{decoupling}, and all joint cumulants of $\mathrm{N}(r_{1}),\ldots,\mathrm{N}(r_{m})$ have non-trivial asymptotics as $n \to + \infty$.


\medskip Let us introduce the following functions:
\begin{align}
& \mathcal{H}_{1}(t; \vec{u},\vec{\mathfrak{s}}):= 1 + \sum_{\ell=1}^{m} \frac{e^{u_{\ell}}-1}{2}\exp\bigg[ \sum_{j=\ell+1}^{m}u_{j} \bigg] \mathrm{erfc}(t-\mathfrak{s}_{\ell}), \label{function H1} \\
& \mathcal{H}_{2}(t; \vec{u},\vec{\mathfrak{s}}):= 1 + \sum_{\ell=1}^{m} \frac{e^{-u_{\ell}}-1}{2}\exp\bigg[ -\sum_{j=1}^{\ell-1}u_{j} \bigg] \mathrm{erfc}(t+\mathfrak{s}_{\ell}), \label{function H2} \\
& \mathcal{G}_{1}(t; \vec{u},\vec{\mathfrak{s}}) := \frac{1}{\mathcal{H}_{1}(t; \vec{u},\vec{\mathfrak{s}})} \sum_{\ell=1}^{m}(e^{u_{\ell}}-1)\exp\bigg[ \sum_{j=\ell+1}^{m}u_{j} \bigg] \frac{e^{-(t-\mathfrak{s}_{\ell})^{2}}}{\sqrt{2\pi}} \frac{1-2\mathfrak{s}_{\ell}^{2}+t\mathfrak{s}_{\ell}-5t^{2}}{3}, \label{function G1} \\
& \mathcal{G}_{2}(t; \vec{u},\vec{\mathfrak{s}}) := \frac{1}{\mathcal{H}_{1}(t; \vec{u},\vec{\mathfrak{s}})} \sum_{\ell=1}^{m} (e^{u_{\ell}}-1)\exp\bigg[ \sum_{j=\ell+1}^{m}u_{j} \bigg] \frac{e^{-(t-\mathfrak{s}_{\ell})^{2}}}{18\sqrt{2\pi}} \bigg( 50t^{5}-70t^{4}\mathfrak{s}_{\ell}-t^{3}\big( 73-62\mathfrak{s}_{\ell}^{2}\big)  \nonumber \\
& +t^{2}\mathfrak{s}_{\ell}\big(33-50\mathfrak{s}_{\ell}^{2}\big) - t \big( 3+18\mathfrak{s}_{\ell}^{2}-16\mathfrak{s}_{\ell}^{4} \big) - \mathfrak{s}_{\ell} \big( 3-22\mathfrak{s}_{\ell}^{2}+8\mathfrak{s}_{\ell}^{4} \big) \bigg), \label{function G2}
\end{align}
where $t \in \mathbb{R}$, $\vec{u}=(u_{1},\ldots,u_{m}) \in \mathbb{C}^{m}$, $\vec{\mathfrak{s}}=(\mathfrak{s}_{1},\ldots,\mathfrak{s}_{m})\in \mathbb{R}^{m}$, and $\mathrm{erfc}$ is the complementary error function
\begin{align}\label{def of erfc}
\mathrm{erfc} (t) = \frac{2}{\sqrt{\pi}}\int_{t}^{\infty} e^{-x^{2}}dx.
\end{align}
The functions $\mathcal{H}_{j}$, $j=1,2$, appear in the denominators of \eqref{function G1}--\eqref{function G2} and inside logarithms in the statements of our main theorems. The next lemma ensures that $\{\mathcal{G}_{j}, \ln \mathcal{H}_{j}\}_{j=1}^{2}$ are well-defined and real-valued for $t \in \mathbb{R}$, $\vec{u}=(u_{1},\ldots,u_{m}) \in \mathbb{R}^{m}$, $\mathfrak{s}_{1}<\ldots<\mathfrak{s}_{m}$. Here and below, $\ln$ always denotes the principal branch of the logarithm. 

\newpage
\begin{lemma}\label{Hjpositivelemma}
$\mathcal{H}_{j}(t; \vec{u},\vec{\mathfrak{s}})>0$ for $j=1,2$ and for all $t \in \mathbb{R}$, $\vec{u}=(u_{1},\ldots,u_{m}) \in \mathbb{R}^{m}$, $\mathfrak{s}_{1}<\ldots<\mathfrak{s}_{m}$. 
\end{lemma}
\begin{proof}
In view of the identity $\mathcal{H}_1(t; \vec{u},\vec{\mathfrak{s}}) = e^{u_{1} + \dots + u_{m}} \mathcal{H}_2(-t; \vec{u},\vec{\mathfrak{s}})$, it is enough to consider $\mathcal{H}_1$. Since
\begin{align*}
\partial_{u_{1}}\mathcal{H}_{1}(t; \vec{u},\vec{\mathfrak{s}}) = \frac{e^{u_{1}+\ldots+u_{m}}}{2}\mathrm{erfc}(t-\mathfrak{s}_{1})>0,
\end{align*}
we only have to check that $\mathcal{H}_{1}|_{u_{1}=-\infty}\geq 0$. It is easy to verify that
\begin{align*}
\mathcal{H}_{1}(t; \vec{u},\vec{\mathfrak{s}})|_{u_{1}=-\infty} = \frac{1}{2}[2-\mathrm{erfc}(t-\mathfrak{s}_{m})] + \sum_{\ell=2}^{m} \frac{e^{u_{\ell}+\ldots+u_{m}}}{2}[\mathrm{erfc}(t-\mathfrak{s}_{\ell})-\mathrm{erfc}(t-\mathfrak{s}_{\ell-1})].
\end{align*}
Since $\mathbb{R}\ni x \mapsto \mathrm{erfc}(x)$ is decreasing from $2$ to $0$ and $\mathfrak{s}_{1}<\ldots<\mathfrak{s}_{m}$, each of the $m$ terms in the above right-hand side is $>0$, which implies $\mathcal{H}_{1}|_{u_{1}=-\infty}>0$.
\end{proof}

The following two theorems are our main results. 
\begin{theorem}\label{thm:main thm}(Merging radii in the bulk)

\noindent Let $m \in \mathbb{N}_{>0}$, $r \in (0,b^{-\frac{1}{2b}})$, $\mathfrak{s}_{1},\ldots,\mathfrak{s}_{m} \in \mathbb{R}$, $\alpha > -1$ and $b>0$ be fixed parameters such that $\mathfrak{s}_{1}<\ldots<\mathfrak{s}_{m}$, and for $n \in \mathbb{N}_{>0}$, define
\begin{align*}
r_{\ell} = r \bigg( 1+\frac{\sqrt{2}\, \mathfrak{s}_{\ell}}{r^{b}\sqrt{n}} \bigg)^{\frac{1}{2b}}, \qquad \ell=1,\ldots,m.
\end{align*}
For any fixed $x_{1},\ldots,x_{m} \in \mathbb{R}$, there exists $\delta > 0$ such that 
\begin{align}\label{asymp in main thm}
\mathbb{E}\bigg[ \prod_{j=1}^{m} e^{u_{j}\mathrm{N}(r_{j})} \bigg] = \exp \bigg( C_{1} n + C_{2} \sqrt{n} + C_{3} +  \frac{C_{4}}{\sqrt{n}} + \bigO\bigg(\frac{(\ln n)^{2}}{n}\bigg)\bigg), \qquad \mbox{as } n \to + \infty
\end{align}
uniformly for $u_{1} \in \{z \in \mathbb{C}: |z-x_{1}|\leq \delta\},\ldots,u_{m} \in \{z \in \mathbb{C}: |z-x_{m}|\leq \delta\}$, where
\begin{align*}
& C_{1} = b r^{2b} \sum_{j=1}^{m}u_{j}, \\
& C_{2} = \sqrt{2}\, b r^{b} \int_{0}^{+\infty} \Big( \ln \mathcal{H}_{1}(t; \vec{u},\vec{\mathfrak{s}}) + \ln \mathcal{H}_{2}(t; \vec{u},\vec{\mathfrak{s}}) \Big) dt, \\
& C_{3} = - \bigg( \frac{1}{2}+\alpha \bigg)\sum_{j=1}^{m}u_{j} + 4b \int_{0}^{+\infty} t\Big( \ln \mathcal{H}_{1}(t; \vec{u},\vec{\mathfrak{s}}) - \ln \mathcal{H}_{2}(t; \vec{u},\vec{\mathfrak{s}}) \Big) dt + \sqrt{2}\, b \int_{-\infty}^{+\infty} \mathcal{G}_{1}(t;\vec{u},\vec{\mathfrak{s}}) dt, \\
& C_{4} = \frac{6\sqrt{2}\, b}{r^{b}} \int_{0}^{+\infty} t^{2}\Big( \ln \mathcal{H}_{1}(t; \vec{u},\vec{\mathfrak{s}}) + \ln \mathcal{H}_{2}(t; \vec{u},\vec{\mathfrak{s}}) \Big) dt \\
& \hspace{0.8cm} + \frac{b}{r^{b}} \int_{-\infty}^{+\infty} \bigg( 4t \, \mathcal{G}_{1}(t; \vec{u},\vec{\mathfrak{s}}) - \frac{\mathcal{G}_{1}(t; \vec{u},\vec{\mathfrak{s}})^{2}}{\sqrt{2}} + \mathcal{G}_{2}(t; \vec{u},\vec{\mathfrak{s}}) \bigg)dt.
\end{align*}
In particular, since $\mathbb{E}\big[ \prod_{j=1}^{m} e^{u_{j}\mathrm{N}(r_{j})} \big]$ depends analytically on $u_{1},\ldots,u_{m} \in \mathbb{C}$ and is positive for $u_{1},\ldots,u_{m} \in \mathbb{R}$, the asymptotic formula \eqref{asymp in main thm} together with Cauchy's formula shows that
\begin{align}\label{der of main result}
\partial_{u_{1}}^{k_{1}}\ldots \partial_{u_{m}}^{k_{m}} \bigg\{ \ln \mathbb{E}\bigg[ \prod_{j=1}^{m} e^{u_{j}\mathrm{N}(r_{j})} \bigg] - \bigg( C_{1} n + C_{2} \sqrt{n} + C_{3} +  \frac{C_{4}}{\sqrt{n}} \bigg) \bigg\} = \bigO\bigg(\frac{(\ln n)^{2}}{n}\bigg), \quad \mbox{as } n \to + \infty,
\end{align}
for any $k_{1},\ldots,k_{m}\in \mathbb{N}$, and $u_{1},\ldots,u_{m}\in \mathbb{R}$. 
\end{theorem}
\begin{theorem}\label{thm:main thm edge}(Merging radii at the edge)

\noindent Let $m \in \mathbb{N}_{>0}$, $\mathfrak{s}_{1},\ldots,\mathfrak{s}_{m} \in \mathbb{R}$, $\alpha > -1$ and $b>0$ be fixed parameters such that $\mathfrak{s}_{1}<\ldots<\mathfrak{s}_{m}$, and for $n \in \mathbb{N}_{>0}$, define
\begin{align*}
r_{\ell} = b^{-\frac{1}{2b}} \bigg( 1+\sqrt{2b}\frac{\mathfrak{s}_{\ell}}{\sqrt{n}} \bigg)^{\frac{1}{2b}}, \qquad \ell=1,\ldots,m.
\end{align*}
For any fixed $x_{1},\ldots,x_{m} \in \mathbb{R}$, there exists $\delta > 0$ such that 
\begin{align}\label{asymp in main thm edge}
\mathbb{E}\bigg[ \prod_{j=1}^{m} e^{u_{j}\mathrm{N}(r_{j})} \bigg] = \exp \bigg( C_{1} n + C_{2} \sqrt{n} + C_{3} +  \frac{C_{4}}{\sqrt{n}} + \bigO\bigg(\frac{(\ln n)^{2}}{n}\bigg)\bigg), \qquad \mbox{as } n \to + \infty
\end{align}
uniformly for $u_{1} \in \{z \in \mathbb{C}: |z-x_{1}|\leq \delta\},\ldots,u_{m} \in \{z \in \mathbb{C}: |z-x_{m}|\leq \delta\}$, where
\begin{align*}
& C_{1} = \sum_{j=1}^{m}u_{j} , \\
& C_{2} = \sqrt{2b} \int_{0}^{+\infty}\ln \mathcal{H}_{2}(t;\vec{u},\vec{\mathfrak{s}})dt , \\
& C_{3} = \bigg( \frac{1}{2}+\alpha \bigg)\ln \mathcal{H}_{2}(0;\vec{u},\vec{\mathfrak{s}}) - 4b \int_{0}^{+\infty} t \ln \mathcal{H}_{2}(t;\vec{u},\vec{\mathfrak{s}})dt +\sqrt{2} \, b \int_{-\infty}^{0}\mathcal{G}_{1}(t;\vec{u},\vec{\mathfrak{s}})dt, \\
& C_{4} = 6\sqrt{2} \; b^{\frac{3}{2}}\int_{0}^{+\infty} t^{2} \ln \mathcal{H}_{2}(t; \vec{u},\vec{\mathfrak{s}}) dt + b^{3/2} \int_{-\infty}^{0} \bigg( 4t \, \mathcal{G}_{1}(t; \vec{u},\vec{\mathfrak{s}}) - \frac{\mathcal{G}_{1}(t; \vec{u},\vec{\mathfrak{s}})^{2}}{\sqrt{2}} + \mathcal{G}_{2}(t; \vec{u},\vec{\mathfrak{s}}) \bigg)dt \\
& \hspace{0.8cm} - \frac{1+6\alpha+6\alpha^{2}}{12\sqrt{2b}}\frac{\mathcal{H}_{2}'(0;\vec{u},\vec{\mathfrak{s}})}{\mathcal{H}_{2}(0;\vec{u},\vec{\mathfrak{s}})} + \bigg(\frac{1}{2}+\alpha\bigg)\sqrt{b} \, \mathcal{G}_{1}(0; \vec{u},\vec{\mathfrak{s}}).
\end{align*}
In particular, since $\mathbb{E}\big[ \prod_{j=1}^{m} e^{u_{j}\mathrm{N}(r_{j})} \big]$ depends analytically on $u_{1},\ldots,u_{m} \in \mathbb{C}$ and is positive for $u_{1},\ldots,u_{m} \in \mathbb{R}$, the asymptotic formula \eqref{asymp in main thm} together with Cauchy's formula shows that
\begin{align}\label{der of main result edge}
\partial_{u_{1}}^{k_{1}}\ldots \partial_{u_{m}}^{k_{m}} \bigg\{ \ln \mathbb{E}\bigg[ \prod_{j=1}^{m} e^{u_{j}\mathrm{N}(r_{j})} \bigg] - \bigg( C_{1} n + C_{2} \sqrt{n} + C_{3} +  \frac{C_{4}}{\sqrt{n}} \bigg) \bigg\} = \bigO\bigg(\frac{(\ln n)^{2}}{n}\bigg), \quad \mbox{as } n \to + \infty,
\end{align}
for any $k_{1},\ldots,k_{m}\in \mathbb{N}$, and $u_{1},\ldots,u_{m}\in \mathbb{R}$. 
\end{theorem}
\begin{remark}
We believe that $\bigO\big(\frac{(\ln n)^{2}}{n}\big)$ in \eqref{asymp in main thm} and \eqref{der of main result} is not optimal and can be improved to $\bigO(n^{-1})$.
\end{remark}

Let $(\mathbb{N}^{m})_{>0} := \{\vec{j}=(j_{1},\ldots,j_{m}) \in \mathbb{N}: j_{1}+\ldots+j_{m}\geq 1\}$. 

Recall that the joint cumulants $\{\kappa_{\vec{j}}=\kappa_{\vec{j}}(r_{1},\ldots,r_{m};n,b,\alpha)\}_{\vec{j} \in (\mathbb{N}^{m})_{>0}}$ of $\mathrm{N}(r_{1}), \ldots, \mathrm{N}(r_{m})$ are defined by
\begin{align}\label{joint cumulant}
\kappa_{\vec{j}}=\kappa_{j_{1},\ldots,j_{m}}:=\partial_{\vec{u}}^{\vec{j}} \ln \mathbb{E}[e^{u_{1}\mathrm{N}(r_{1})+\ldots + u_{m}\mathrm{N}(r_{m})}] \Big|_{\vec{u}=\vec{0}}, \qquad \vec{j} \in (\mathbb{N}^{m})_{>0},
\end{align}
where $\partial_{\vec{u}}^{\vec{j}}:=\partial_{u_{1}}^{j_{1}}\ldots \partial_{u_{m}}^{j_{m}}$ and $\vec{0}:=(0,\ldots,0)$. In particular, we have
\begin{align*}
\mathbb{E}[\mathrm{N}(r)] = \kappa_{1}(r), \qquad \mbox{Var}[\mathrm{N}(r)] = \kappa_{2}(r), \qquad \mbox{Cov}[\mathrm{N}(r_{1}),\mathrm{N}(r_{2})] = \kappa_{(1,1)}(r_{1},r_{2}).
\end{align*}
Corollary \ref{coro:correlation} below follows from Theorems \ref{thm:main thm} and \ref{thm:main thm edge}. As already mentioned, it contains the large $n$ asymptotics of all joint cumulants of $\mathrm{N}(r_{1}),\ldots,\mathrm{N}(r_{m})$ when the radii are merging, either in the bulk or at the edge, and also contains several central limit theorems for the joint fluctuations of $\mathrm{N}(r_{1}),\ldots,\mathrm{N}(r_{m})$.

\begin{corollary}\label{coro:correlation}
(a) (bulk regime) Let $m \in \mathbb{N}_{>0}$, $\vec{j} \in (\mathbb{N}^{m})_{>0}$, $\alpha > -1$, $b>0$, $r \in (0,b^{-\frac{1}{2b}})$ and $-\infty<\mathfrak{s}_{1}<\ldots <\mathfrak{s}_{m}<+\infty$ be fixed, and for $n \in \mathbb{N}_{>0}$, define
\begin{align*}
r_{\ell} = r \bigg( 1+\frac{\sqrt{2}\, \mathfrak{s}_{\ell}}{r^{b}\sqrt{n}} \bigg)^{\frac{1}{2b}}, \qquad \ell=1,\ldots,m.
\end{align*}
As $n \to +\infty$, we have
\begin{align}\label{asymp cumulant bulk}
\kappa_{\vec{j}} = \partial_{\vec{u}}^{\vec{j}}C_{1}\big|_{\vec{u}=\vec{0}} n + \partial_{\vec{u}}^{\vec{j}}C_{2}\big|_{\vec{u}=\vec{0}} \sqrt{n} + \partial_{\vec{u}}^{\vec{j}}C_{3}\big|_{\vec{u}=\vec{0}} +  \frac{\partial_{\vec{u}}^{\vec{j}}C_{4}\big|_{\vec{u}=\vec{0}}}{\sqrt{n}} + \bigO \bigg( \frac{(\ln n)^{2}}{n} \bigg),
\end{align}
where $C_{1},\ldots,C_{4}$ are as in Theorem \ref{thm:main thm}. In particular, for any $1 \leq \ell < k \leq m$, as $n \to + \infty$ we have
\begin{align}
& \mathbb{E}[\mathrm{N}(r_{\ell})] = br^{2b}n + \sqrt{2}\, br^{b}\mathfrak{s}_{\ell}\sqrt{n} + \frac{b-1-2\alpha}{2} + \bigO\bigg( \frac{(\ln n)^{2}}{n} \bigg), \label{exp bulk coro} \\
& \mathrm{Var}[\mathrm{N}(r_{\ell})] = \frac{br^{b}}{\sqrt{\pi}}\sqrt{n} + \frac{b \, \mathfrak{s}_{\ell}}{\sqrt{2\pi}} - \frac{b(1+4\mathfrak{s}_{\ell}^{2})}{16\sqrt{\pi}r^{b}}\frac{1}{\sqrt{n}} + \bigO\bigg( \frac{(\ln n)^{2}}{n} \bigg), \label{var bulk coro} \\
& \mathrm{Cov}(\mathrm{N}(r_{\ell}),\mathrm{N}(r_{k})) = c_{(1,1)}(\mathfrak{s}_{\ell},\mathfrak{s}_{k})\sqrt{n} + d_{(1,1)}(\mathfrak{s}_{\ell},\mathfrak{s}_{k}) + e_{(1,1)}(\mathfrak{s}_{\ell},\mathfrak{s}_{k})n^{-\frac{1}{2}} + \bigO \bigg( \frac{(\ln n)^{2}}{n} \bigg), \nonumber
\end{align}
where
\begin{align}
& c_{(1,1)}(\mathfrak{s}_{\ell},\mathfrak{s}_{k}) \hspace{-0.03cm} = \hspace{-0.03cm} \frac{b r^{b}}{\sqrt{2}} \int_{0}^{+\infty} \hspace{-0.1cm} \Big\{ \mathrm{erfc}(t-\mathfrak{s}_{\ell}) \big( 1- \tfrac{1}{2}\mathrm{erfc}(t-\mathfrak{s}_{k})\big) + \mathrm{erfc}(t+\mathfrak{s}_{k}) \big( 1- \tfrac{1}{2}\mathrm{erfc}(t+\mathfrak{s}_{\ell}) \big) \Big\} dt, \label{def of c11 bulk} \\
& d_{(1,1)}(\mathfrak{s}_{\ell},\mathfrak{s}_{k}) = b \int_{0}^{+\infty} t \Big\{ \mathrm{erfc}(t-\mathfrak{s}_{\ell}) \big( 2-\mathrm{erfc}(t-\mathfrak{s}_{k})\big) - \mathrm{erfc}(t+\mathfrak{s}_{k}) \big( 2-\mathrm{erfc}(t+\mathfrak{s}_{\ell}) \big) \Big\}dt \nonumber \\
& + b \int_{-\infty}^{+\infty} \bigg\{ \big( 2-\mathrm{erfc}(t-\mathfrak{s}_{k})\big) \frac{e^{-(t-\mathfrak{s}_{\ell})^{2}}}{2\sqrt{\pi}} \frac{1-5t^{2}+t\mathfrak{s}_{\ell}-2\mathfrak{s}_{\ell}^{2}}{3} 
	\nonumber \\ 
& \hspace{1.7cm} - \mathrm{erfc}(t-\mathfrak{s}_{\ell}) \frac{e^{-(t-\mathfrak{s}_{k})^{2}}}{2\sqrt{\pi}} \frac{1-5t^{2}+t\mathfrak{s}_{k}-2\mathfrak{s}_{k}^{2}}{3} \bigg\}dt, \nonumber \\
& e_{(1,1)}(\mathfrak{s}_{\ell},\mathfrak{s}_{k}) = -\frac{br^{-b}e^{-\frac{(\mathfrak{s}_{\ell}-\mathfrak{s}_{k})^{2}}{2}}}{288\sqrt{\pi}}  \Big( 51+55\mathfrak{s}_{\ell}^{4}+55\mathfrak{s}_{k}^{4} + 96 \mathfrak{s}_{\ell}^{2} + 96 \mathfrak{s}_{k}^{2} + 128 \mathfrak{s}_{\ell}^{3}\mathfrak{s}_{k} + 128 \mathfrak{s}_{\ell}\mathfrak{s}_{k}^{3} + 180 \mathfrak{s}_{\ell}\mathfrak{s}_{k}  \nonumber \\ 
& + 210 \mathfrak{s}_{\ell}^{2}\mathfrak{s}_{k}^{2} \Big) + \frac{3b r^{-b}}{\sqrt{2}} \int_{0}^{+\infty} t^{2} \Big\{ \mathrm{erfc}(t-\mathfrak{s}_{\ell}) \big(2-\mathrm{erfc}(t-\mathfrak{s}_{k})\big) + \big(2-\mathrm{erfc}(t+\mathfrak{s}_{\ell})\big)\mathrm{erfc}(t+\mathfrak{s}_{k})   \Big\}dt \nonumber \\
& + \frac{b r^{-b}}{36\sqrt{2\pi}} \int_{-\infty}^{+\infty} \bigg[ \big( 2-\mathrm{erfc}(t-\mathfrak{s}_{k}) \big) e^{-(t-\mathfrak{s}_{\ell})^{2}}\mathfrak{p}(t,\mathfrak{s}_{\ell}) - \mathrm{erfc}(t-\mathfrak{s}_{\ell}) e^{-(t-\mathfrak{s}_{k})^{2}}\mathfrak{p}(t,\mathfrak{s}_{k}) \bigg]dt, \nonumber \\[0.1cm]
& \mathfrak{p}(t,\mathfrak{s}) = -3\mathfrak{s}+22\mathfrak{s}^{3}-8\mathfrak{s}^{5} + t(21-66\mathfrak{s}^{2}+16\mathfrak{s}^{4}) + t^{2}(57\mathfrak{s}-50 \mathfrak{s}^{3}) + t^{3}(-193+62\mathfrak{s}^{2}) -70 t^{4} \mathfrak{s} + 50 t^{5}. \nonumber
\end{align}
(b) (joint fluctuations in the bulk) Let $\alpha > -1$, $b>0$, $r \in (0,b^{-\frac{1}{2b}})$, $m \in \mathbb{N}_{>0}$ and $\mathfrak{s}_{1},\ldots,\mathfrak{s}_{m} \in \mathbb{R}$ be fixed, and for $n \in \mathbb{N}_{>0}$, define
\begin{align*}
r_{\ell} = r \bigg( 1+\frac{\sqrt{2}\, \mathfrak{s}_{\ell}}{r^{b}\sqrt{n}} \bigg)^{\frac{1}{2b}}, \qquad \ell=1,\ldots,m.
\end{align*} Consider the random variables
\begin{align}
& \mathcal{N}_{\ell} := \pi^{1/4}\frac{\mathrm{N}(r_{\ell})-(br^{2b}n+\sqrt{2} \, br^{b} \mathfrak{s}_{\ell} \sqrt{n})}{\sqrt{br^{b}} \; n^{1/4}}, \qquad \ell=1,\ldots,m. \label{Nj bulk}
\end{align}
As $n \to + \infty$, $(\mathcal{N}_{1},\ldots,\mathcal{N}_{m})$ convergences in distribution to a multivariate normal random variable of mean $(0,\ldots,0)$ and whose covariance matrix $\Sigma$ is given by
\begin{align*}
\Sigma_{\ell,\ell} = 1, \qquad \Sigma_{\ell,k} = \Sigma_{k,\ell} = \frac{ c_{(1,1)}(\mathfrak{s}_{\ell},\mathfrak{s}_{k})}{br^{b}/\sqrt{\pi}}, \qquad 1 \leq \ell<k \leq m,
\end{align*}
where $c_{(1,1)}(\mathfrak{s}_{\ell},\mathfrak{s}_{k})$ is given by \eqref{def of c11 bulk}. \\[-0.2cm]

\noindent (c) (edge regime) Let $m \in \mathbb{N}_{>0}$, $\vec{j} \in (\mathbb{N}^{m})_{>0}$, $\alpha > -1$, $b>0$ and $-\infty < \mathfrak{s}_{1}< \ldots<\mathfrak{s}_{m}< +\infty$ be fixed, and for $n \in \mathbb{N}_{>0}$, define
\begin{align*}
r_{\ell} = b^{-\frac{1}{2b}} \bigg( 1+\frac{\sqrt{2b}\, \mathfrak{s}_{\ell}}{\sqrt{n}} \bigg)^{\frac{1}{2b}}, \qquad \ell=1,\ldots,m.
\end{align*}
As $n \to +\infty$, we have
\begin{align}\label{asymp cumulant edge}
\kappa_{\vec{j}} = \partial_{\vec{u}}^{\vec{j}}C_{1}\big|_{\vec{u}=\vec{0}} n + \partial_{\vec{u}}^{\vec{j}}C_{2}\big|_{\vec{u}=\vec{0}} \sqrt{n} + \partial_{\vec{u}}^{\vec{j}}C_{3}\big|_{\vec{u}=\vec{0}} +  \frac{\partial_{\vec{u}}^{\vec{j}}C_{4}\big|_{\vec{u}=\vec{0}}}{\sqrt{n}} + \bigO \bigg( \frac{(\ln n)^{2}}{n} \bigg),
\end{align}
where $C_{1},\ldots,C_{4}$ are as in Theorem \ref{thm:main thm edge}. In particular, for any $1 \leq \ell < k \leq m$, as $n \to + \infty$ we have
\begin{align*}
& \mathbb{E}[\mathrm{N}(r_{\ell})] = n + c_{1}(\mathfrak{s}_{\ell})\sqrt{n} + d_{1}(\mathfrak{s}_{\ell}) + e_{1}(\mathfrak{s}_{\ell}) n^{-\frac{1}{2}} + \bigO\bigg( \frac{(\ln n)^{2}}{n} \bigg), \\
& \mathrm{Var}[\mathrm{N}(r_{\ell})] = c_{2}(\mathfrak{s}_{\ell})\sqrt{n} + d_{2}(\mathfrak{s}_{\ell}) + e_{2}(\mathfrak{s}_{\ell}) n^{-\frac{1}{2}} + \bigO\bigg( \frac{(\ln n)^{2}}{n} \bigg), \\
& \mathrm{Cov}(\mathrm{N}(r_{\ell}),\mathrm{N}(r_{k})) = c_{(1,1)}(\mathfrak{s}_{\ell},\mathfrak{s}_{k})\sqrt{n} + d_{(1,1)}(\mathfrak{s}_{\ell},\mathfrak{s}_{k}) + e_{(1,1)}(\mathfrak{s}_{\ell},\mathfrak{s}_{k}) n^{-\frac{1}{2}} + \bigO \bigg( \frac{(\ln n)^{2}}{n} \bigg),
\end{align*}
where
\begin{align}
& c_{1}(\mathfrak{s}) = \frac{\sqrt{b} \, \mathfrak{s}}{\sqrt{2}}\mathrm{erfc}(\mathfrak{s}) - \frac{\sqrt{b}}{\sqrt{2\pi}}e^{-\mathfrak{s}^{2}}, \label{c1 edge} \\
& d_{1}(\mathfrak{s}) = -\frac{1}{2}\bigg( \frac{1}{2}+\alpha - \frac{b}{2} \bigg) \mathrm{erfc} (\mathfrak{s}) - \frac{b \, \mathfrak{s}}{3\sqrt{\pi}}e^{-\mathfrak{s}^{2}}, \nonumber \\
& e_{1}(\mathfrak{s}) = \frac{e^{-\mathfrak{s}^{2}}}{\sqrt{2\pi}}\bigg( \frac{b(2+4\alpha)-1-6\alpha-6\alpha^{2}}{12\sqrt{b}} + \frac{(3b-2-4\alpha)\mathfrak{s}^{2}}{6}\sqrt{b} - \frac{2\mathfrak{s}^{4}}{9}b^{3/2} \bigg), \nonumber \\
& c_{2}(\mathfrak{s}) = \frac{\sqrt{b}}{2\sqrt{\pi}}\mathrm{erfc}(\sqrt{2} \, \mathfrak{s}) + \sqrt{b}\frac{e^{-\mathfrak{s}^{2}}}{\sqrt{2\pi}}\big( 1-\mathrm{erfc} (\mathfrak{s})\big) + \frac{\sqrt{b} \, \mathfrak{s}}{\sqrt{2}}\mathrm{erfc} ( \mathfrak{s}) \bigg( \frac{1}{2}\mathrm{erfc}( \mathfrak{s})-1 \bigg), \label{c2 edge} \\
& d_{2}(\mathfrak{s}) = -\frac{b}{12 \pi}e^{-2 \mathfrak{s}^{2}} + \frac{b \,\mathfrak{s}}{2\sqrt{2\pi}}\mathrm{erfc}( \sqrt{2} \,  \mathfrak{s} ) + \frac{b \, \mathfrak{s}}{3\sqrt{\pi}}e^{-\mathfrak{s}^{2}}\big( 1-\mathrm{erfc}( \mathfrak{s}) \big) \nonumber \\
& \hspace{1.15cm} + \frac{b-1-2\alpha}{4} \mathrm{erfc}( \mathfrak{s}) \bigg( \frac{1}{2}\mathrm{erfc}( \mathfrak{s}) - 1 \bigg), \nonumber \\
& e_{2}(\mathfrak{s}) = \frac{e^{-\mathfrak{s}^{2}}}{12\sqrt{2\pi b}} \bigg( 1-2b+6\alpha-4b \alpha + 6\alpha^{2} + 2(2-3b+4\alpha) b \,\mathfrak{s}^{2} + \frac{8b^{2}}{3}\mathfrak{s}^{4} \bigg) \big( 1-\mathrm{erfc}( \mathfrak{s} ) \big) \nonumber \\
& \hspace{0.75cm}  -\frac{b^{3/2}  \mathfrak{s}}{72 \sqrt{2} \, \pi}e^{-2\mathfrak{s}^{2}} - \frac{b^{3/2}(1+4\mathfrak{s}^{2})}{32\sqrt{\pi}}\mathrm{erfc}( \sqrt{2} \, \mathfrak{s} ), \nonumber \\
& c_{(1,1)}(\mathfrak{s}_{\ell},\mathfrak{s}_{k}) = \frac{\sqrt{b}}{2\sqrt{2}} \int_{0}^{+\infty} \mathrm{erfc}(t+\mathfrak{s}_{k}) \big( 2- \mathrm{erfc}(t+\mathfrak{s}_{\ell}) \big) dt, \label{c11 edge} \\
& d_{(1,1)}(\mathfrak{s}_{\ell},\mathfrak{s}_{k}) = \frac{1+2\alpha}{8}\big( 2-\mathrm{erfc}(\mathfrak{s}_{\ell}) \big) \mathrm{erfc}(\mathfrak{s}_{k}) -  b \int_{0}^{+\infty} t \; \mathrm{erfc}(t+\mathfrak{s}_{k}) \big( 2-\mathrm{erfc}(t+\mathfrak{s}_{\ell}) \big) dt \nonumber \\
& + b \int_{-\infty}^{0} \bigg\{ \big( 2-\mathrm{erfc}(t-\mathfrak{s}_{k})\big) \frac{e^{-(t-\mathfrak{s}_{\ell})^{2}}}{2\sqrt{\pi}} \frac{1-5t^{2}+t\mathfrak{s}_{\ell}-2\mathfrak{s}_{\ell}^{2}}{3} \nonumber \\
&\hspace{1.65cm} - \mathrm{erfc}(t-\mathfrak{s}_{\ell}) \frac{e^{-(t-\mathfrak{s}_{k})^{2}}}{2\sqrt{\pi}} \frac{1-5t^{2}+t\mathfrak{s}_{k}-2\mathfrak{s}_{k}^{2}}{3} \bigg\}dt, \nonumber \\
& e_{(1,1)}(\mathfrak{s}_{\ell},\mathfrak{s}_{k}) = (2-\mathrm{erfc}(\mathfrak{s}_{\ell})) \frac{e^{-\mathfrak{s}_{k}^{2}}}{\sqrt{2\pi}} \frac{1+6\alpha+6\alpha^{2}+2b(1+2\alpha)(2\mathfrak{s}_{k}^{2}-1)}{24\sqrt{b}} \nonumber \\
& \hspace{+1.82cm} - \mathrm{erfc}(\mathfrak{s}_{k}) \frac{e^{-\mathfrak{s}_{\ell}^{2}}}{\sqrt{2\pi}} \frac{1+6\alpha+6\alpha^{2}+2b(1+2\alpha)(2\mathfrak{s}_{\ell}^{2}-1)}{24\sqrt{b}} \nonumber \\
& -\frac{b^{\frac{3}{2}}e^{-\frac{(\mathfrak{s}_{\ell}-\mathfrak{s}_{k})^{2}}{2}}}{288\sqrt{\pi}} \frac{1}{2}\mathrm{erfc}\bigg( \frac{\mathfrak{s}_{\ell}+\mathfrak{s}_{k}}{\sqrt{2}} \bigg) \Big( 51+55\mathfrak{s}_{\ell}^{4}+55\mathfrak{s}_{k}^{4} + 96 \mathfrak{s}_{\ell}^{2} + 96 \mathfrak{s}_{k}^{2} + 128 \mathfrak{s}_{\ell}^{3}\mathfrak{s}_{k} + 128 \mathfrak{s}_{\ell}\mathfrak{s}_{k}^{3}  \nonumber \\ 
& + 180 \mathfrak{s}_{\ell}\mathfrak{s}_{k} + 210 \mathfrak{s}_{\ell}^{2}\mathfrak{s}_{k}^{2} \Big) + \frac{b^{\frac{3}{2}}}{144\sqrt{2}}\frac{e^{-\mathfrak{s}_{\ell}^{2}-\mathfrak{s}_{k}^{2}}}{2\pi} \Big( 55(\mathfrak{s}_{\ell}^{3}+\mathfrak{s}_{k}^{3}) + 73 (\mathfrak{s}_{\ell}+\mathfrak{s}_{k} + \mathfrak{s}_{\ell}^{2}\mathfrak{s}_{k} + \mathfrak{s}_{\ell}\mathfrak{s}_{k}^{2}) \Big) \nonumber \\
& + \frac{3b^{\frac{3}{2}}}{\sqrt{2}} \int_{0}^{+\infty} t^{2} \big(2-\mathrm{erfc}(t+\mathfrak{s}_{\ell})\big)\mathrm{erfc}(t+\mathfrak{s}_{k}) dt \nonumber \\
& + \frac{b^{\frac{3}{2}}}{36\sqrt{2\pi}} \int_{-\infty}^{0} \bigg[ \big( 2-\mathrm{erfc}(t-\mathfrak{s}_{k}) \big) e^{-(t-\mathfrak{s}_{\ell})^{2}}\mathfrak{p}(t,\mathfrak{s}_{\ell}) - \mathrm{erfc}(t-\mathfrak{s}_{\ell}) e^{-(t-\mathfrak{s}_{k})^{2}}\mathfrak{p}(t,\mathfrak{s}_{k}) \bigg]dt, \nonumber \\[0.1cm]
& \mathfrak{p}(t,\mathfrak{s}) = -3\mathfrak{s}+22\mathfrak{s}^{3}
-8\mathfrak{s}^{5} + t(21-66\mathfrak{s}^{2}+16\mathfrak{s}^{4}) + t^{2}(57\mathfrak{s}-50 \mathfrak{s}^{3}) + t^{3}(-193+62\mathfrak{s}^{2}) -70 t^{4} \mathfrak{s} + 50 t^{5}. \nonumber
\end{align}
(d) (joint fluctuations at the edge) Let $\alpha > -1$, $b>0$, $m \in \mathbb{N}_{>0}$ and $\mathfrak{s}_{1},\ldots,\mathfrak{s}_{m} \in \mathbb{R}$ be fixed, and for $n \in \mathbb{N}_{>0}$, define
\begin{align*}
r_{\ell} = b^{-\frac{1}{2b}} \bigg( 1+\sqrt{2b}\frac{\mathfrak{s}_{\ell}}{\sqrt{n}} \bigg)^{\frac{1}{2b}}, \qquad \ell=1,\ldots,m.
\end{align*}
Consider the random variables
\begin{align}
& \mathcal{N}_{\ell} := \frac{\mathrm{N}(r_{\ell})-(n+c_{1}( \mathfrak{s}_{\ell})\sqrt{n})}{\sqrt{c_{2}(\mathfrak{s}_{\ell})} \; n^{1/4}}, \qquad \ell=1,\ldots,m, \label{Nj edge}
\end{align}
where $c_{1}$ is given by \eqref{c1 edge} and $c_{2}$ is given by \eqref{c2 edge}.
As $n \to + \infty$, $(\mathcal{N}_{1},\ldots,\mathcal{N}_{m})$ convergences in distribution to a multivariate normal random variable of mean $(0,\ldots,0)$ and whose covariance matrix $\Sigma$ is given by
\begin{align*}
\Sigma_{\ell,\ell} = 1, \qquad \Sigma_{\ell,k} = \Sigma_{k,\ell} = \frac{ c_{(1,1)}(\mathfrak{s}_{\ell},\mathfrak{s}_{k})}{\sqrt{c_{2}(\mathfrak{s}_{\ell})}\sqrt{c_{2}(\mathfrak{s}_{k})}}, \qquad 1 \leq \ell<k \leq m,
\end{align*}
where $c_{(1,1)}(\mathfrak{s}_{\ell},\mathfrak{s}_{k})$ is given by \eqref{def of c11 bulk}.
\end{corollary}
\begin{remark}
Some of the results contained in Corollary \ref{coro:correlation} were already known:
\begin{itemize}
\item In \cite[eq (70)]{LeeRiser2016}, Lee and Riser obtained second-order asymptotics for the number of points lying outside the droplet of the Ellictic Ginibre ensemble (in particular, the coefficients $c_{1}(0)|_{(b,\alpha)=(1,0)}$ and $d_{1}(0)|_{(b,\alpha)=(1,0)}$ of part (c) above were contained in their results). 
\item Given a Borel set $A$, let $N_{A} := \#\{z_{j}:z_{j}\in A\}$. In \cite{CE2020}, Charles and Estienne proved that if $A$ is independent of $n$, has smooth boundary and lies strictly in the bulk of the Ginibre ensemble, the cumulants $\{\kappa_{j}(A)\}_{j=1}^{+\infty}$ enjoy an all-order expansion of the form
\begin{align}\label{all order expansion}
\kappa_{j}(A) = \begin{cases}
\alpha_{j,0}n \hspace{+0.2cm} + \sum_{k=1}^{N} \alpha_{j,k}n^{1-k} + \bigO(n^{-N}), & \mbox{if } j =1, \\
\hspace{1.43cm} \sum_{k=1}^{N} \alpha_{j,k}n^{1-k} + \bigO(n^{-N}), & \mbox{if $j$ is odd and } j \geq 3, \\
\beta_{j,0}n^{\frac{1}{2}} + \sum_{k=1}^{N} \beta_{j,k}n^{\frac{1}{2}-k} + \bigO(n^{-N-\frac{1}{2}}), & \mbox{if $j$ is even,}
\end{cases}
\end{align}
where $N \in \mathbb{N}$ is arbitrary. Furthermore, the coefficients $\alpha_{j,0}$ and $\beta_{j,0}$ were computed explicitly. It can be verified (see \cite[Corollary 1.4 (a)]{Charlier 2d jumps}) that \eqref{asymp cumulant bulk} is consistent with \eqref{all order expansion}: for $m=1$ and $\mathfrak{s}_{1}=0$, we have $\partial_{u}^{j}C_{1}\big|_{u=0}=0$ for $j \geq 2$, $\partial_{u}^{j}C_{2}\big|_{u=0}=0=\partial_{u}^{j}C_{4}\big|_{u=0}=0$ for $j$ odd, and $\partial_{u}^{j}C_{3}\big|_{u=0}=0$ for $j$ even.
\item Second order asymptotics for the cumulants $\{\kappa_{j}\}_{j=1}^{+\infty}$ (i.e. the case $m=1$) were obtained in \cite[eqs (55)--(67)]{LMS2018} for general $b>0$ and $\alpha>-1$, both in the bulk and the edge regimes.
\item Third order asymptotics for the cumulants $\{\kappa_{j}\}_{j=1}^{+\infty}$ (i.e. the case $m=1$) were obtained in \cite[Remark 4]{FenzlLambert} for the Ginibre case, and the leading coefficient $c_{(1,1)}(\mathfrak{s}_{\ell},\mathfrak{s}_{k})|_{(b,\alpha)=(1,0)}$ (given in \eqref{def of c11 bulk} for the bulk and in \eqref{c11 edge} for the edge) was obtained in \cite[Proposition 2.3]{FenzlLambert}. Parts (b) and (d) above, when specialized to $(b,\alpha)=(1,0)$, were also already known from \cite[Proposition 2.3]{FenzlLambert}.
\item Part (a) with $m=1$ and $\mathfrak{s}_{1}=0$ and part (c) with $m=1$ and general $\mathfrak{s}_{1} \in \mathbb{R}$ were already known from \cite[Corollary 1.4]{Charlier 2d jumps}.
\end{itemize}
\end{remark}
\begin{proof}[Proof of Corollary \ref{coro:correlation}]
Proof of parts $(a)$ and $(c)$: The asymptotics \eqref{asymp cumulant bulk} and \eqref{asymp cumulant edge} directly follow from \eqref{der of main result}, \eqref{der of main result edge} and \eqref{joint cumulant}. The simplified asymptotics for $\mathbb{E}[\mathrm{N}(r_{\ell})]$, $\mathrm{Var}[\mathrm{N}(r_{\ell})]$, and $\mathrm{Cov}(\mathrm{N}(r_{\ell}),\mathrm{N}(r_{k}))$ are then obtained by performing a long but straightforward computation. Proof of part $(d)$: Let $t_{1},\ldots,t_{m}\in \mathbb{R}$ be arbitrary but fixed. Note that $c_{2}(\mathfrak{s}) > 0$ for $\mathfrak{s} \in \mathbb{R}$, because $c_{2}'(\mathfrak{s})  = 2^{-3/2}\sqrt{b}(\mathrm{erfc}(\mathfrak{s}) - 2)\mathrm{erfc}(\mathfrak{s}) < 0$ and $\lim_{\mathfrak{s}\to+\infty} c_{2}(\mathfrak{s}) = 0$. By Theorem \ref{thm:main thm edge}, we know that \eqref{asymp in main thm edge} holds uniformly for $u_{1},\ldots,u_{m}\in \{z \in \mathbb{C}:|z|\leq \delta\}$ for a certain $\delta>0$. Hence, using Theorem \ref{thm:main thm edge} with
\begin{align*}
u_{\ell} = \frac{t_{\ell}}{\sqrt{c_{2}(\mathfrak{s}_{\ell})} \, n^{1/4}}, \qquad \ell=1,\ldots,m,
\end{align*}
we obtain
\begin{align*}
\mathbb{E}\bigg[ \prod_{\ell=1}^{m}e^{t_{\ell}\mathcal{N}_{\ell}} \bigg] = \exp \bigg( \frac{1}{2}\sum_{1 \leq \ell < k \leq m} \Sigma_{\ell,k}t_{\ell}t_{k} + \bigO(n^{-\frac{1}{4}}) \bigg), \qquad \mbox{as } n \to +\infty.
\end{align*}
Thus the above asymptotics imply pointwise convergence in $(t_{1},\ldots,t_{m})\in \mathbb{R}^{m}$ of $\mathbb{E}\big[ \prod_{\ell=1}^{m}e^{t_{\ell}\mathcal{N}_{\ell}} \big]$ to $\exp \big( \frac{1}{2}\sum_{1 \leq \ell < k \leq m} \Sigma_{\ell,k}t_{\ell}t_{k} \big)$ as $n \to + \infty$. This, in turn, implies by standard theorems that $(\mathcal{N}_{1},\ldots,\mathcal{N}_{m})$ convergences in distribution to a multivariate normal random variable of mean $\vec{0}$ and covariance matrix $\Sigma$, which finishes the proof of (d). The proof of (b) is similar.
\end{proof}

\textbf{Determinants.} Here we express $\mathbb{E}\big[ \prod_{\ell=1}^{m} e^{u_{\ell}\mathrm{N}(r_{\ell})} \big]$ as a ratio of two determinants. Using that $\prod_{1 \leq j < k \leq n} |z_{k} -z_{j}|^{2}$ is the product of two Vandermonde determinants, we obtain after standard manipulations that
\begin{align}
\mathbb{E}\bigg[ \prod_{\ell=1}^{m} e^{u_{\ell}\mathrm{N}(r_{\ell})} \bigg] & = \frac{1}{n!Z_{n}} \int_{\mathbb{C}}\ldots \int_{\mathbb{C}} \prod_{1 \leq j < k \leq n} |z_{k} -z_{j}|^{2} \prod_{j=1}^{n}w(z_{j}) d^{2}z_{j} \nonumber \\
& = \frac{1}{Z_{n}} \det \left( \int_{\mathbb{C}} z^{j} \overline{z}^{k} w(z) d^{2}z \right)_{j,k=0}^{n-1} \label{def of Dn as n fold integral} \\
& = \frac{1}{Z_{n}}(2\pi)^{n}\prod_{j=0}^{n-1}\int_{0}^{+\infty}u^{2j+1}w(u)du, \label{simplified determinant}
\end{align}
where the weight $w$ is defined by
\begin{align}\label{def of w and omega}
w(z):=|z|^{2\alpha} e^{-n |z|^{2b}} \omega(|z|), \qquad \omega(x) := \prod_{\ell=1}^{m} \begin{cases}
e^{u_{\ell}}, & \mbox{if } x < r_{\ell}, \\
1, & \mbox{if } x \geq r_{\ell}.
\end{cases}
\end{align}
Formula \eqref{simplified determinant} directly follows from \eqref{def of Dn as n fold integral} and the fact that $w$ is rotation-invariant. Indeed, since $w(z)=w(|z|)$, the integral $\int_{\mathbb{C}} z^{j} \overline{z}^{k} w(z) d^{2}z$ is $0$ for $j \neq k$ and is $2\pi  \int_{0}^{+\infty}u^{2j+1}w(u)du$ for $j=k$. So only the main diagonal contributes for the determinants in \eqref{def of Dn as n fold integral}. 

\medskip \textbf{Related works.} We note from \eqref{def of Dn as n fold integral} that the problem of determining the large $n$ asymptotics of $\mathbb{E}\big[ \prod_{j=1}^{m} e^{u_{j}\mathrm{N}(r_{j})} \big]$ can equivalently be seen as a problem of obtaining large $n$ asymptotics for an $n \times n$ determinant whose weight is supported on $\mathbb{C}$, rotation-invariant, and with $m$ merging discontinuities along circles. For Theorem \ref{thm:main thm}, the discontinuities are merging in the bulk, while for Theorem \ref{thm:main thm edge} the discontinuities are merging at the edge.

\medskip The one-dimensional analogue of this merging of discontinuities has been studied by several authors in the context of Toeplitz, Hankel and Toeplitz+Hankel determinants. Large $n$ asymptotics of $n \times n$ Toeplitz determinants with \textit{two} merging discontinuities were first obtained in the important works \cite{CIK2011, CK2015}. In both \cite{CIK2011} and \cite{CK2015}, the term of order $1$ in the asymptotics is characterized in terms of the solution to a Painlev\'{e} V equation. The generalization of \cite{CK2015} to the case where an arbitrary number of discontinuities are merging is a challenging problem and only recently important progress has been made \cite{Fahs}. Toeplitz+Hankel determinants with merging singularities have also recently been studied in \cite{FK2021, CGMY2020}, and some applications of these results are given in \cite{CFK2022}. In the aforementioned works, the discontinuities are merging in the bulk. Hankel determinants with merging discontinuities at the edge are also related to the Painlev\'{e} theory, see \cite{WuXu2021, LC2021} for soft edges, \cite{LCX2022} for a hard edge, and e.g. \cite{CD2018, CD2019, ChCl3} for studies on related Fredholm determinants. It is interesting to note that, in contrast to these works, the asymptotics obtained in Theorems \ref{thm:main thm} and \ref{thm:main thm edge} for merging circular discontinuities do not involve transcendental functions.

\medskip Let us also discuss related results on determinants with singularities in a two-dimensional setting. The works \cite{CE2020, L et al 2019, FenzlLambert, Charlier 2d jumps} were already mentioned at the beginning of the introduction and deal with determinants having (non-merging) discontinuities. Beyond determinants with discontinuities, determinants with root-type singularities and related planar orthogonal polynomials have also attracted considerable attention in recent years \cite{BBLM2015, BGM17, LeeYang, WebbWong, BEG18, LeeYang2, LeeYang3, DeanoSimm}. The analogues of Theorems \ref{thm:main thm} and \ref{thm:main thm edge} for planar root-type singularities can be found in \cite[Theorem 1.5]{DeanoSimm} (two merging singularities in the bulk) and \cite[Theorem 1.14]{DeanoSimm} (an arbitrary number of merging singularities at the edge).



\section{Proof of Theorem \ref{thm:main thm}}\label{section:proof}
Recall that $\omega$ was defined in \eqref{def of w and omega}. For convenience, let us rewrite it as
\begin{align}\label{def of omegaell}
\omega (x) = \sum_{\ell=1}^{m+1}\omega_{\ell} \mathbf{1}_{[0,r_{\ell})}(x) = \sum_{\ell=1}^{m+1}\Omega_{\ell} \mathbf{1}_{[r_{\ell-1},r_{\ell})}(x),
\end{align}
where $r_{m+1}:=+\infty$ and
\begin{align}\label{def of Omega j}
\omega_{\ell} := \begin{cases}
e^{u_{\ell}+\ldots+u_{m}}-e^{u_{\ell+1}+\ldots+u_{m}}, & \mbox{if } \ell < m, \\
e^{u_{m}}-1, & \mbox{if } \ell=m, \\
1, & \mbox{if } \ell=m+1,
\end{cases} \quad \Omega_{\ell} = \sum_{j=\ell}^{m+1}\omega_{j} = \begin{cases}
e^{u_{\ell}+\ldots+u_{m}}, & \mbox{if } \ell \leq m, \\
1 & \mbox{if } \ell=m+1.
\end{cases}
\end{align}
The starting point of our proof is the following formula:
\begin{align}\label{main exact formula}
\ln \mathcal{E}_{n} = \sum_{j=1}^{n} \ln \bigg(1+\sum_{\ell=1}^{m} \omega_{\ell} \frac{\gamma(\tfrac{j+\alpha}{b},nr_{\ell}^{2b})}{\Gamma(\tfrac{j+\alpha}{b})} \bigg),
\end{align}
where $\mathcal{E}_{n}:=\mathbb{E}\big[ \prod_{\ell=1}^{m} e^{u_{\ell}\mathrm{N}(r_{\ell})} \big]$ and $\gamma(a,z)$ is the incomplete gamma function
\begin{align*}
\gamma(a,z) = \int_{0}^{z}t^{a-1}e^{-t}dt.
\end{align*}
The identity \eqref{main exact formula} can easily be derived from \eqref{simplified determinant} and was also obtained in \cite[eqs (1.23) and (1.26)]{Charlier 2d jumps}. We infer from \eqref{main exact formula} that the asymptotics of $\gamma(a,z)$ as $z \to +\infty$ uniformly for $a\in [\frac{1+\alpha}{b},\frac{z}{b r_{1}^{2b}}+\frac{\alpha}{b}]$ are needed to obtain large $n$ asymptotics for $\mathcal{E}_{n}$ --- we recall these asymptotics in Appendix \ref{section:uniform asymp gamma}. 

\medskip In \eqref{main exact formula} and below, $\ln$ always denotes the principal branch of the logarithm. 

\medskip Our proof strategy follows \cite{Charlier 2d jumps}. Let us define
\begin{align*}
& j_{-}:=\lceil \tfrac{bnr^{2b}}{1+\epsilon} - \alpha \rceil, \qquad j_{+} := \lfloor  \tfrac{bnr^{2b}}{1-\epsilon} - \alpha \rfloor,
\end{align*}
where $\epsilon > 0$ is independent of $n$. We assume that $\epsilon$ is sufficiently small such that
\begin{align*}
\frac{br^{2b}}{1-\epsilon} < \frac{1}{1+\epsilon},
\end{align*}
so that we can write
\begin{align}\label{log Dn as a sum of sums}
\ln \mathcal{E}_{n} = S_{0} + S_{1} + S_{2} + S_{3},
\end{align}
where
\begin{align}
& S_{0} = \sum_{j=1}^{M'} \ln \bigg( 1+\sum_{\ell=1}^{m} \omega_{\ell} \frac{\gamma(\tfrac{j+\alpha}{b},nr_{\ell}^{2b})}{\Gamma(\tfrac{j+\alpha}{b})} \bigg), & & S_{1} = \sum_{j=M'+1}^{j_{-}-1} \hspace{-0.3cm} \ln \bigg( 1+\sum_{\ell=1}^{m} \omega_{\ell} \frac{\gamma(\tfrac{j+\alpha}{b},nr_{\ell}^{2b})}{\Gamma(\tfrac{j+\alpha}{b})} \bigg), \label{def of S0 and S1} \\
& S_{2} = \sum_{j=j_{-}}^{j_{+}} \ln \bigg( 1+\sum_{\ell=1}^{m} \omega_{\ell} \frac{\gamma(\tfrac{j+\alpha}{b},nr_{\ell}^{2b})}{\Gamma(\tfrac{j+\alpha}{b})} \bigg), & & S_{3}=\sum_{j=j_{+}+1}^{n} \hspace{-0.3cm} \ln \bigg( 1+\sum_{\ell=1}^{m} \omega_{\ell} \frac{\gamma(\tfrac{j+\alpha}{b},nr_{\ell}^{2b})}{\Gamma(\tfrac{j+\alpha}{b})} \bigg). \label{def of S2 and S3}
\end{align}
In the above, $M'>0$ is an integer independent of $n$. For $j=1,\ldots,n$ and $k =1,\ldots,m$, we also define $a_{j}:=\frac{j+\alpha}{b}$, $z_{k}:=nr_{k}^{2b}$ and
\begin{align}\label{def etajl}
\lambda_{j,k} := \frac{z_{k}}{a_{j}} = \frac{bnr_{k}^{2b}}{j+\alpha}, \qquad \lambda_{j} := \frac{bnr^{2b}}{j+\alpha}, \qquad \eta_{j,k} := (\lambda_{j,k}-1)\sqrt{\frac{2 (\lambda_{j,k}-1-\ln \lambda_{j,k})}{(\lambda_{j,k}-1)^{2}}}.
\end{align}
\begin{lemma}\label{lemma: S0}
For any $x_{1},\ldots,x_{m} \in \mathbb{R}$, there exists $\delta > 0$ such that
\begin{align}\label{asymp of S0}
S_{0} = M' \ln \Omega_{1} + \bigO(e^{-cn}), \qquad \mbox{as } n \to + \infty,
\end{align}
uniformly for $u_{1} \in \{z \in \mathbb{C}: |z-x_{1}|\leq \delta\},\ldots,u_{m} \in \{z \in \mathbb{C}: |z-x_{m}|\leq \delta\}$.
\end{lemma}
\begin{proof}
We infer from \eqref{def of S0 and S1} and Lemma \ref{lemma:various regime of gamma} that
\begin{align*}
S_{0} & = \sum_{j=1}^{M'} \ln \bigg( \sum_{\ell=1}^{m+1} \omega_{\ell} \big[1 + \bigO(e^{-cn}) \big] \bigg) = \sum_{j=1}^{M'} \ln \Omega_{1} + \bigO(e^{-cn}), \quad \mbox{as } n \to +\infty.
\end{align*}
In the above, the error terms before the second equality are independent of $u_{1},\ldots,u_{m}$, so the claim follows.
\end{proof}
\begin{lemma}\label{lemma: S2km1}
The constant $M'$ can be chosen sufficiently large such that the following holds. For any $x_{1},\ldots,x_{m} \in \mathbb{R}$, there exists $\delta > 0$ such that 
\begin{align*}
& S_{1} = (j_{-}-M'-1) \ln \Omega_{1} + \bigO(e^{-cn}), \qquad S_{3} = \bigO(e^{-cn}),
\end{align*}
as $n \to +\infty$ uniformly for $u_{1} \in \{z \in \mathbb{C}: |z-x_{1}|\leq \delta\},\ldots,u_{m} \in \{z \in \mathbb{C}: |z-x_{m}|\leq \delta\}$.
\end{lemma}
\begin{proof}
The proof is identical to the proof of \cite[Lemma 2.2]{Charlier 2d jumps}, so we omit it.
\end{proof}
We now focus on $S_{2}$. Let $M:=M'\sqrt{\ln n}$. For the analysis we need to split $S_{2}$ as follows
\begin{align}\label{asymp prelim of S2kpvp}
& S_{2}=S_{2}^{(1)}+S_{2}^{(2)}+S_{2}^{(3)}, \qquad S_{2}^{(v)} := \sum_{j:\lambda_{j}\in I_{v}}  \ln \bigg( 1+\sum_{\ell=1}^{m} \omega_{\ell} \frac{\gamma(\tfrac{j+\alpha}{b},nr_{\ell}^{2b})}{\Gamma(\tfrac{j+\alpha}{b})} \bigg), \quad v=1,2,3,
\end{align}
where
\begin{align*}
I_{1} = [1-\epsilon,1-\tfrac{M}{\sqrt{n}}), \qquad I_{2} = [1-\tfrac{M}{\sqrt{n}},1+\tfrac{M}{\sqrt{n}}], \qquad I_{3} = (1+\tfrac{M}{\sqrt{n}},1+\epsilon].
\end{align*}
From \eqref{asymp prelim of S2kpvp}, we see that the large $n$ asymptotics of $\{S_{2}^{(v)}\}_{v=1,2,3}$ involve the asymptotics of $\gamma(a,z)$ when $a \to + \infty$, $z \to +\infty$ with $\lambda=\frac{z}{a} \in [1-\epsilon,1+\epsilon]$. These sums can also be rewritten using
\begin{align}\label{sums lambda j}
& \sum_{j:\lambda_{j}\in I_{3}} = \sum_{j=j_{-}}^{g_{-}-1}, \qquad \sum_{j:\lambda_{j}\in I_{2}} = \sum_{j= g_{-}}^{g_{+}}, \qquad \sum_{j:\lambda_{j}\in I_{1}} = \sum_{j= g_{+}+1}^{j_{+}},
\end{align}
where $g_{-} := \lceil \frac{bnr^{2b}}{1+\frac{M}{\sqrt{n}}}-\alpha \rceil$, $g_{+} := \lfloor \frac{bnr^{2b}}{1-\frac{M}{\sqrt{n}}}-\alpha \rfloor$.

\medskip It turns out that $S_{2}^{(1)}$, $S_{2}^{(2)}$ and $S_{2}^{(3)}$ have oscillatory asymptotics as $n \to + \infty$. To handle these oscillations, we follow \cite{Charlier 2d jumps} and introduce the following quantities:
\begin{align*}
& \theta_{-}^{(n,M)} := g_{-} - \bigg( \frac{bn r^{2b}}{1+\frac{M}{\sqrt{n}}} - \alpha \bigg) = \bigg\lceil \frac{bn r^{2b}}{1+\frac{M}{\sqrt{n}}} - \alpha \bigg\rceil - \bigg( \frac{bn r^{2b}}{1+\frac{M}{\sqrt{n}}} - \alpha \bigg), \\
& \theta_{+}^{(n,M)} := \bigg( \frac{bn r^{2b}}{1-\frac{M}{\sqrt{n}}} - \alpha \bigg) - g_{+} = \bigg( \frac{bn r^{2b}}{1-\frac{M}{\sqrt{n}}} - \alpha \bigg) - \bigg\lfloor \frac{bn r^{2b}}{1-\frac{M}{\sqrt{n}}} - \alpha \bigg\rfloor.
\end{align*}
Note that $\theta_{-}^{(n,M)},\theta_{+}^{(n,M)} \in [0,1)$ are oscillatory but remain bounded as $n \to + \infty$.

\begin{lemma}\label{lemma:S2kp1p}
The constant $M'$ can be chosen sufficiently large such that the following holds. For any $x_{1},\ldots,x_{m} \in \mathbb{R}$, there exists $\delta > 0$ such that
\begin{align*}
S_{2}^{(1)} = \bigO(n^{-10}),
\end{align*}
as $n \to +\infty$ uniformly for $u_{1} \in \{z \in \mathbb{C}: |z-x_{1}|\leq \delta\},\ldots,u_{m} \in \{z \in \mathbb{C}: |z-x_{m}|\leq \delta\}$. 
\end{lemma}
\begin{proof}
Using \eqref{asymp prelim of S2kpvp} and Lemma \ref{lemma: uniform}, we get
\begin{align*}
S_{2}^{(1)} & = \sum_{j:\lambda_{j}\in I_{1}} \ln \bigg( 1+\sum_{\ell=1}^{m} \omega_{\ell} \bigg[ \frac{1}{2}\mathrm{erfc}\Big(-\eta_{j,\ell} \sqrt{a_{j}/2}\Big) - R_{a_{j}}(\eta_{j,\ell}) \bigg] \bigg).
\end{align*}
Furthermore, for all sufficiently large $n$ we have
\begin{align}
& \eta_{j,\ell} = \lambda_{j,\ell}-1 +\bigO((\lambda_{j,\ell}-1)^{2}) \leq -\tfrac{M}{\sqrt{n}} + \bigO(\tfrac{1}{\sqrt{n}}), & & -\eta_{j,\ell} \sqrt{a_{j}/2} \geq  \tfrac{Mr^{b}}{\sqrt{2}} + \bigO(1), \label{lol11}
\end{align}
uniformly for $j\in \{j: \lambda_{j} \in I_{1}\}$. Hence, for sufficiently large $M'$ we have
\begin{align*}
& R_{a_{j}}(\eta_{j,\ell}) =  \bigO(e^{-\frac{r^{2b}M^{2}}{4}}) = \bigO(n^{-11}), & & \frac{1}{2}\mathrm{erfc}\Big(-\eta_{j,\ell} \sqrt{a_{j}/2}\Big) = \bigO(e^{-\frac{r^{2b}M^{2}}{4}})= \bigO(n^{-11}),
\end{align*}
as $n \to + \infty$ uniformly for $j\in \{j: \lambda_{j} \in I_{1}\}$. Thus, by \eqref{sums lambda j},
\begin{align}
S_{2}^{(1)} = \bigO(n^{-10}), \qquad \mbox{as } n \to + \infty. \label{lol12}
\end{align}
The error terms in \eqref{lol11} are independent of $\omega_{1},\ldots,\omega_{m}$, and therefore the error term in \eqref{lol12} is uniform for $u_{1} \in \{z \in \mathbb{C}: |z-x_{1}|\leq \delta\},\ldots,u_{m} \in \{z \in \mathbb{C}: |z-x_{m}|\leq \delta\}$. 
\end{proof}

\begin{lemma}\label{lemma:S2kp3p}
The constant $M'$ can be chosen sufficiently large such that the following holds. For any $x_{1},\ldots,x_{m} \in \mathbb{R}$, there exists $\delta > 0$ such that
\begin{align*}
S_{2}^{(3)} = & \; \Big( br^{2b}n - j_{-} - bMr^{2b}\sqrt{n} + bM^{2}r^{2b} -\alpha+\theta_{-}^{(n,M)} - bM^{3}r^{2b}n^{-\frac{1}{2}} \Big) \ln  \Omega_{1} + \bigO(M^{4}n^{-1}),
\end{align*}
as $n \to +\infty$ uniformly for $u_{1} \in \{z \in \mathbb{C}: |z-x_{1}|\leq \delta\},\ldots,u_{m} \in \{z \in \mathbb{C}: |z-x_{m}|\leq \delta\}$.
\end{lemma}
\begin{proof}
The claim can be proved in a similar way as Lemma \ref{lemma:S2kp1p}. Using \eqref{asymp prelim of S2kpvp} and Lemma \ref{lemma: uniform}, we obtain
\begin{align*}
S_{2}^{(3)} & = \sum_{j:\lambda_{j}\in I_{3}} \ln \bigg( 1+\sum_{\ell=1}^{m} \omega_{\ell} \bigg[ \frac{1}{2}\mathrm{erfc}\Big(-\eta_{j,\ell} \sqrt{a_{j}/2}\Big) - R_{a_{j}}(\eta_{j,\ell}) \bigg] \bigg).
\end{align*}
Since for all sufficiently large $n$ we have
\begin{align*}
& \eta_{j,\ell} = \lambda_{j,\ell}-1 +\bigO((\lambda_{j,\ell}-1)^{2}) \geq \tfrac{M}{\sqrt{n}} + \bigO(\tfrac{1}{\sqrt{n}}), & & -\eta_{j,\ell} \sqrt{a_{j}/2} \leq - \tfrac{M r^{b}}{\sqrt{2}} + \bigO(1),
\end{align*}
uniformly for $j\in \{j:\lambda_{j}\in I_{3}\}$, we can choose $M'$ large enough such that
\begin{align*}
& R_{a_{j}}(\eta_{j,\ell}) = \bigO(e^{-\frac{r^{2b}M^{2}}{4}}) = \bigO(n^{-10}), & & \frac{1}{2}\mathrm{erfc}\Big(-\eta_{j,\ell} \sqrt{a_{j}/2}\Big) = 1-\bigO(e^{-\frac{r^{2b}M^{2}}{4}}) = 1-\bigO(n^{-10}),
\end{align*}
as $n \to + \infty$ uniformly for $j\in \{j:\lambda_{j}\in I_{3}\}$, and thus, by \eqref{sums lambda j},
\begin{align*}
S_{2}^{(3)} & = \sum_{j=j_{-}}^{g_{-}-1} \ln \Omega_{1} +\bigO(n^{-9}) = (g_{-}-j_{-})  \ln  \Omega_{1} +\bigO(n^{-9}).
\end{align*}
The claim now follows from
\begin{align*}
\sum_{j = j_{-}}^{g_{-}-1} 1 & = g_{-}-j_{-} = \bigg( \frac{bn r^{2b}}{1+\frac{M}{\sqrt{n}}} - \alpha \bigg)+\theta_{-}^{(n,M)} -j_{-} \nonumber \\
& = br^{2b}n - j_{-} - bMr^{2b}\sqrt{n} + bM^{2}r^{2b}-\alpha+\theta_{-}^{(n,M)} - bM^{3}r^{2b}n^{-\frac{1}{2}} + \bigO(M^{4}n^{-1}), \quad \mbox{as } n \to + \infty.
\end{align*} 
\end{proof}
We now turn to the asymptotic analysis of $S_{2}^{(2)}$. For all $k \in \{1,\ldots,m\}$ and $j \in \{j: \lambda_{j} \in I_{2}\}=\{g_{-},\ldots,g_{+}\}$, define
\begin{align*}
& M_{j,k} := \sqrt{n}(\lambda_{j,k}-1), \qquad M_{j} := \sqrt{n}(\lambda_{j}-1).
\end{align*}
Note that $M_{j,k}$ and $M_{j}$ decrease as $j$ increases. Since $I_{2} = [1-\tfrac{M}{\sqrt{n}},1+\tfrac{M}{\sqrt{n}}]$, as $n \to +\infty$ the points $M_{g_{-},k}, \ldots, M_{g_{+},k}$ run over the interval 
\begin{align*}
\big[\sqrt{n}\big((\tfrac{r_{k}}{r})^{2b}(1-\tfrac{M}{\sqrt{n}})-1\big),\sqrt{n}\big((\tfrac{r_{k}}{r})^{2b}(1+\tfrac{M}{\sqrt{n}})-1\big)\big] \approx [-M+\sqrt{2}r^{-b} \, \mathfrak{s}_{k},M+\sqrt{2}r^{-b} \, \mathfrak{s}_{k}]
\end{align*}
for each $k \in \{1,\ldots,m\}$, and the points $M_{g_{-}}, \ldots, M_{g_{+}}$ run over the interval $[-M,M]$. For the large $n$ asymptotics of $\smash{S_{2}^{(2)}}$ we will need the following lemma.
\begin{lemma}\label{lemma:Riemann sum}(Taken from \cite[Lemma 2.7]{Charlier 2d jumps})
Let $f \in C^{3}(\mathbb{R})$ be a function such that $|f|,|f'|,|f''|,|f'''|$ are bounded. As $n \to + \infty$ we have
\begin{align}
& \sum_{j=g_{-}}^{g_{+}}f(M_{j}) = br^{2b} \int_{-M}^{M} f(t) dt \; \sqrt{n} - 2 b r^{2b} \int_{-M}^{M} tf(t) dt + \bigg( \frac{1}{2}-\theta_{-}^{(n,M)} \bigg)f(M)+ \bigg( \frac{1}{2}-\theta_{+}^{(n,M)} \bigg)f(-M) \nonumber \\
& + \frac{1}{\sqrt{n}}\bigg[ 3br^{2b} \int_{-M}^{M}t^{2}f(t)dt + \bigg( \frac{1}{12}+\frac{\theta_{-}^{(n,M)}(\theta_{-}^{(n,M)}-1)}{2} \bigg)\frac{f'(M)}{br^{2b}} - \bigg( \frac{1}{12}+\frac{\theta_{+}^{(n,M)}(\theta_{+}^{(n,M)}-1)}{2} \bigg)\frac{f'(-M)}{br^{2b}} \bigg] \nonumber \\
& + \bigO(M^{4}n^{-1}). \label{sum f asymp 2}
\end{align}
\end{lemma} 

\begin{lemma}\label{lemma:S2kp2p}
For any $x_{1},\ldots,x_{p} \in \mathbb{R}$, there exists $\delta > 0$ such that
\begin{align*}
&  S_{2}^{(2)} = \widetilde{C}_{2}^{(M)}\sqrt{n} + \widetilde{C}_{3}^{(n,M)} + \frac{1}{\sqrt{n}}\widetilde{C}_{4}^{(n,M)} + \bigO(M^{4}n^{-1}), \\
& \widetilde{C}_{2}^{(M)} = br^{2b} \int_{-M}^{M} f_{1}(t) dt,  \\
& \widetilde{C}_{3}^{(n,M)} = b r^{2b} \int_{-M}^{M} \Big( -2t f_{1}(t)+f_{2}(t) \Big) dt + \bigg( \frac{1}{2}-\theta_{-}^{(n,M)} \bigg)f_{1}(M)+\bigg( \frac{1}{2}-\theta_{+}^{(n,M)} \bigg)f_{1}(-M), \\
& \widetilde{C}_{4}^{(n,M)} = b r^{2b} \int_{-M}^{M}\Big( 3t^{2}f_{1}(t)-2tf_{2}(t)+f_{3}(t) \Big) dt + \bigg( \frac{1}{2}-\theta_{-}^{(n,M)} \bigg)f_{2}(M)+\bigg( \frac{1}{2}-\theta_{+}^{(n,M)} \bigg)f_{2}(-M) \\
& + \bigg( \frac{1}{12}+\frac{\theta_{-}^{(n,M)}(\theta_{-}^{(n,M)}-1)}{2} \bigg) \frac{f_{1}'(M)}{br^{2b}} - \bigg( \frac{1}{12}+\frac{\theta_{+}^{(n,M)}(\theta_{+}^{(n,M)}-1)}{2} \bigg) \frac{f_{1}'(-M)}{br^{2b}} ,
\end{align*}
as $n \to +\infty$ uniformly for $u_{1} \in \{z \in \mathbb{C}: |z-x_{1}|\leq \delta\},\ldots,u_{p} \in \{z \in \mathbb{C}: |z-x_{p}|\leq \delta\}$, where
\begin{align*}
& g(x) = 1+\sum_{\ell=1}^{m} \frac{\omega_{\ell}}{2}\mathrm{erfc}\bigg( -\frac{r^{b}}{\sqrt{2}}(x+\sqrt{2}\,r^{-b} \mathfrak{s}_{\ell}) \bigg), \qquad f_{1}(x) = \ln(g(x)), \\
& f_{2}(x) = \frac{1}{g(x)}\sum_{\ell=1}^{m} \frac{-\omega_{\ell} e^{-\frac{(x+\sqrt{2} r^{-b} \mathfrak{s}_{\ell})^{2}r^{2b}}{2}}}{6r^{b}\sqrt{2\pi}} \bigg( 5r^{2b}x^{2}+\sqrt{2}r^{b} x \mathfrak{s}_{\ell} + 4 \mathfrak{s}_{\ell}^{2} - 2 \bigg), \\
& f_{3}(x) = -\frac{f_{2}(x)^{2}}{2} + \frac{1}{g(x)}\sum_{\ell=1}^{m} \frac{\omega_{\ell}e^{-\frac{(x+\sqrt{2}r^{-b} \mathfrak{s}_{\ell})^{2}r^{2b}}{2}}}{72r^{2b}\sqrt{2\pi}} \bigg\{ -25r^{5b}x^{5} - 35\sqrt{2} \, r^{4b}x^{4}\mathfrak{s}_{\ell}+ r^{3b}x^{3}(73-62\mathfrak{s}_{\ell}^{2}) \\
& +\sqrt{2} \, r^{2b}x^{2}\mathfrak{s}_{\ell}(33-50\mathfrak{s}_{\ell}^{2})+2r^{b}x(3+18\mathfrak{s}_{\ell}^{2}-16\mathfrak{s}_{\ell}^{4}) - 2\sqrt{2}\mathfrak{s}_{\ell}(3-22\mathfrak{s}_{\ell}^{2}+8\mathfrak{s}_{\ell}^{4}) \bigg\}.
\end{align*}
\end{lemma}
\begin{proof}
Using \eqref{asymp prelim of S2kpvp} and Lemma \ref{lemma: uniform}, we obtain
\begin{align}\label{lol1}
& S_{2}^{(2)} = \sum_{j:\lambda_{j}\in I_{2}} \ln \bigg( 1+\sum_{\ell=1}^{m} \omega_{\ell} \bigg[ \frac{1}{2}\mathrm{erfc}\Big(-\eta_{j,\ell} \sqrt{a_{j}/2}\Big) - R_{a_{j}}(\eta_{j,\ell}) \bigg] \bigg).
\end{align}
For $j \in \{j:\lambda_{j}\in I_{2}\}$, we have $1-\frac{M}{\sqrt{n}} \leq \lambda_{j} = \frac{bnr^{2b}}{j+\alpha} \leq 1+\frac{M}{\sqrt{n}}$, $-M \leq M_{j} \leq M$, and 
\begin{align*}
M_{j,k} = M_{j} + \sqrt{2} \, r^{-b} \mathfrak{s}_{k} + \frac{\sqrt{2} \, r^{-b}\mathfrak{s}_{k}M_{j}}{\sqrt{n}}, \qquad k=1,\ldots,m.
\end{align*}
Furthermore, as $n \to + \infty$ we have 
\begin{align}
\eta_{j,\ell} & = (\lambda_{j,\ell}-1)\bigg( 1 - \frac{\lambda_{j,\ell}-1}{3} + \frac{7}{36}(\lambda_{j,\ell}-1)^{2} +\bigO((\lambda_{j,\ell}-1)^{3})\bigg) = \frac{M_{j,\ell}}{\sqrt{n}} - \frac{M_{j,\ell}^{2}}{3n} + \frac{7M_{j,\ell}^{3}}{36n^{3/2}} + \bigO\bigg(\frac{M^{4}}{n^{2}}\bigg), \nonumber \\
& = \frac{M_{j}+\sqrt{2}\, r^{-b} \mathfrak{s}_{\ell}}{\sqrt{n}} + \frac{1}{3n} \bigg( \sqrt{2} \,r^{-b} s_{\ell}M_{j} - M_{j}^{2} - 2r^{-2b} \mathfrak{s}_{\ell}^{2} \bigg) \nonumber \\
&  + \frac{1}{3n^{3/2}} \big( M_{j}+\sqrt{2}\, r^{-b} \mathfrak{s}_{\ell} \big)\bigg( \frac{7}{12}\big( M_{j}+\sqrt{2}\, r^{-b}\mathfrak{s}_{\ell} \big)^{2} - 2 \sqrt{2} \,r^{-b} \mathfrak{s}_{\ell} M_{j} \bigg) + \bigO\bigg(\frac{M^{4}}{n^{2}}\bigg) \label{asymp etaj and etajsqrtajover2 1}
\end{align}
and
\begin{align} 
-\eta_{j,\ell} \sqrt{a_{j}/2} & = - \frac{M_{j,\ell} r_{\ell}^{b}}{\sqrt{2}} + \frac{5M_{j,\ell}^{2} r_{\ell}^{b}}{6\sqrt{2}\sqrt{n}} - \frac{53 M_{j,\ell}^{3} r_{\ell}^{b}}{72\sqrt{2}n} +\bigO(M^{4}n^{-\frac{3}{2}}) \nonumber \\
& = - \frac{(M_{j}+\sqrt{2} \,r^{-b} \mathfrak{s}_{\ell})r^{b}}{\sqrt{2}} + \frac{r^{b}}{12\sqrt{n}}\bigg( 5\sqrt{2} M_{j}^{2} + 2r^{-b} M_{j} \mathfrak{s}_{\ell} + 4\sqrt{2} r^{-2b} \mathfrak{s}_{\ell}^{2} \bigg) \nonumber \\
& - \frac{r^{b}}{144n}\bigg( 53\sqrt{2}M_{j}^{3} + 18r^{-b} M_{j}^{2}\mathfrak{s}_{\ell} + 12\sqrt{2} \, r^{-2b} M_{j} \mathfrak{s}_{\ell}^{2} + 56r^{-3b} \mathfrak{s}_{\ell}^{3} \bigg) + \bigO(M^{4}n^{-3/2}), \label{asymp etaj and etajsqrtajover2 2} 
\end{align}
uniformly for $j\in \{j:\lambda_{j}\in I_{2}\}$. Hence, by \eqref{asymp of Ra}, as $n \to + \infty$ we have
\begin{align*}
& R_{a_{j}}(\eta_{j,\ell}) = \frac{e^{-\frac{(M_{j}+\sqrt{2}r^{-b} \mathfrak{s}_{\ell})^{2}r^{2b}}{2}}}{\sqrt{2\pi}} \bigg( \frac{-1}{3r^{b}\sqrt{n}} \nonumber \\
&  - \frac{10M_{j}^{3}r^{3b}+12\sqrt{2}M_{j}^{2}r^{2b}\mathfrak{s}_{\ell}+12 M_{j} r^{b} \mathfrak{s}_{\ell}^{2}+8\sqrt{2}\mathfrak{s}_{\ell}^{3}+3M_{j}r^{b}-3\sqrt{2}\mathfrak{s}_{\ell}}{36r^{2b}n} + \bigO((1+M_{j}^{6})n^{-\frac{3}{2}}) \bigg)
\end{align*}
and
\begin{align*}
& \frac{1}{2}\mathrm{erfc}\Big(-\eta_{j,\ell} \sqrt{a_{j}/2}\Big) = \frac{1}{2}\mathrm{erfc}\Big(-\frac{r^{b}}{\sqrt{2}}(M_{j}+\sqrt{2}\, r^{-b} \mathfrak{s}_{\ell})\Big) \\
& -\frac{e^{-\frac{(M_{j}+\sqrt{2}r^{-b} \mathfrak{s}_{\ell})^{2}r^{2b}}{2}}}{12\sqrt{\pi}r^{b}\sqrt{n}}\Big( 5\sqrt{2}r^{2b}M_{j}^{2}+2r^{b} M_{j} \mathfrak{s}_{\ell} + 4\sqrt{2} \mathfrak{s}_{\ell}^{2} \Big) \\
& + \frac{e^{-\frac{(M_{j}+\sqrt{2}r^{-b} \mathfrak{s}_{\ell})^{2}r^{2b}}{2}}}{144\sqrt{\pi}r^{2b}n}\bigg\{ 53\sqrt{2}r^{3b}M_{j}^{3} + 18r^{2b} M_{j}^{2}\mathfrak{s}_{\ell} + 12\sqrt{2} r^{b} M_{j} \mathfrak{s}_{\ell}^{2} + 56 \mathfrak{s}_{\ell}^{3} \\
& - \sqrt{2} \Big( r^{b}M_{j}+\sqrt{2} \mathfrak{s}_{\ell} \Big) \Big( 5r^{2b}M_{j}^{2}+\sqrt{2} r^{b}M_{j}\mathfrak{s}_{\ell} + 4 \mathfrak{s}_{\ell}^{2} \Big)^{2} \bigg\} + \bigO\Big(e^{-\frac{(M_{j}+\sqrt{2}r^{-b} \mathfrak{s}_{\ell})^{2}r^{2b}}{2}}(1+M_{j}^{8}) n^{-\frac{3}{2}}\Big),
\end{align*}
uniformly for $j\in \{j:\lambda_{j}\in I_{2}\}$. 
The identity $g(- \sqrt{2} r^{-b} t) = \mathcal{H}_{1}(t; \vec{u},\vec{\mathfrak{s}})$ in combination with Lemma \ref{Hjpositivelemma} shows that $g(x) > 0$ for all $x \in \mathbb{R}$; in particular, the functions $f_1(x),f_{2}(x),f_{3}(x)$ are well-defined and real-valued for $x \in \mathbb{R}$. Substituting the above asymptotics into \eqref{lol1} and using that the error terms are suppressed by exponentials of the form $e^{-c M_j^2}$, we obtain
\begin{align}
S_{2}^{(2)}  & = \Sigma_{1}^{(n)}+\frac{1}{\sqrt{n}}\Sigma_{2}^{(n)}+\frac{1}{n}\Sigma_{3}^{(n)} + \bigO(n^{-1}), \qquad \mbox{as } n \to + \infty, \label{asymp of S2kp2p in proof}
\end{align}
where
\begin{align*}
& \Sigma_{1}^{(n)} = \sum_{j=g_{-}}^{g_{+}} f_{1}(M_{j}), \qquad \Sigma_{2}^{(n)} = \sum_{j=g_{-}}^{g_{+}} f_{2}(M_{j}), \qquad \Sigma_{3}^{(n)} = \sum_{j=g_{-}}^{g_{+}} f_{3}(M_{j}).
\end{align*}
The functions $f_j$, $j = 1,2,3$, satisfy the assumptions of Lemma \ref{lemma:Riemann sum}. Moreover, $f_2(x)$, $f_3(x)$, and their derivatives have exponential decay as $x \to \pm \infty$. Hence, by \eqref{sum f asymp 2}, we have
\begin{align*}
& \Sigma_{1}^{(n)} = \Sigma_{1,2} \sqrt{n} + \Sigma_{1,3} + \frac{1}{\sqrt{n}} \Sigma_{1,4} + \bigO(M^{4}n^{-1}), \\
& \frac{1}{\sqrt{n}}\Sigma_{2}^{(n)} = \Sigma_{2,3} + \frac{1}{\sqrt{n}} \Sigma_{2,4} + \bigO(n^{-1}), \qquad \frac{1}{n}\Sigma_{3}^{(n)} = \frac{1}{\sqrt{n}} \Sigma_{3,4} + \bigO(n^{-1}),
\end{align*}
as $n \to + \infty$, for some explicit $\Sigma_{1,2},\Sigma_{1,3},\Sigma_{1,4},\Sigma_{2,3},\Sigma_{2,4},\Sigma_{3,4}$. A computation gives
\begin{align*}
& \Sigma_{1,2} =  \widetilde{C}_{2}^{(M)}, & & \Sigma_{1,3} + \Sigma_{2,3} = \widetilde{C}_{3}^{(n,M)}, & & \Sigma_{1,4} + \Sigma_{2,4} + \Sigma_{3,4} = \widetilde{C}_{4}^{(n,M)},
\end{align*}
which is the claim.
\end{proof}
We are now ready to derive the asymptotics of $S_{2}$ as $n \to + \infty$.
\begin{lemma}\label{lemma: asymp of S2k final}
The constant $M'$ can be chosen sufficiently large such that the following holds. For any $x_{1},\ldots,x_{p} \in \mathbb{R}$, there exists $\delta > 0$ such that
\begin{align*}
& S_{2} =  \Big( br^{2b}n - j_{-} \Big) \ln  \Omega_{1}  + C_{2} \sqrt{n} + \widetilde{C}_{3} + \frac{1}{\sqrt{n}} C_{4} + \bigO(M^{4}n^{-1}),
\end{align*}
as $n \to +\infty$ uniformly for $u_{1} \in \{z \in \mathbb{C}: |z-x_{1}|\leq \delta\},\ldots,u_{p} \in \{z \in \mathbb{C}: |z-x_{p}|\leq \delta\}$, where
\begin{align*}
& C_{2} = b r^{2b} \bigg[ \int_{-\infty}^{0}f_{1}(t)dt + \int_{0}^{+\infty} \big( f_{1}(t)-\ln \Omega_{1} \big) dt \bigg], \\
& \widetilde{C}_{3} = \bigg( \frac{1}{2} - \alpha \bigg) \ln \Omega_{1} + b r^{2b} \bigg[ \int_{-\infty}^{0}(-2tf_{1}(t)+f_{2}(t))dt + \int_{0}^{+\infty} (-2t[f_{1}(t)-\ln \Omega_{1}]+f_{2}(t)) dt \bigg], \\
& C_{4} = br^{2b} \bigg[ \int_{-\infty}^{0} \big( 3t^{2}f_{1}(t)-2tf_{2}(t)+f_{3}(t) \big)dt + \int_{0}^{+\infty} \big( 3t^{2}(f_{1}(t)-\ln(\Omega_{1}))-2tf_{2}(t)+f_{3}(t) \big)dt \bigg].
\end{align*}
\end{lemma}
\begin{proof}
It follows from Lemmas \ref{lemma:S2kp1p}, \ref{lemma:S2kp3p} and \ref{lemma:S2kp2p} that
\begin{align*}
& S_{2} = \Big( br^{2b}n - j_{-} \Big) \ln  \Omega_{1} + \widehat{C}_{2}^{(M)} \sqrt{n} + \widehat{C}_{3}^{(n,M)} + \frac{1}{\sqrt{n}} \widehat{C}_{4}^{(n,M)} + \bigO(M^{4}n^{-1}),
\end{align*}
as $n \to +\infty$ uniformly for $u_{1} \in \{z \in \mathbb{C}: |z-x_{1}|\leq \delta\},\ldots,u_{p} \in \{z \in \mathbb{C}: |z-x_{p}|\leq \delta\}$, where
\begin{align*}
& \widehat{C}_{2}^{(M)} := \widetilde{C}_{2}^{(M)} - bM r^{2b} \ln \Omega_{1}, \\
& \widehat{C}_{3}^{(n,M)} := \widetilde{C}_{3}^{(n,M)} + \Big( bM^{2}r^{2b} -\alpha+\theta_{-}^{(n,M)}  \Big) \ln \Omega_{1}, \\
& \widehat{C}_{4}^{(n,M)} := \widetilde{C}_{4}^{(n,M)} - bM^{3}r^{2b} \ln \Omega_{1}.
\end{align*}
Provided $M'$ is chosen sufficiently large, as $n \to + \infty$ we get
\begin{align*}
\widehat{C}_{2}^{(M)} = C_{2} + \bigO(n^{-10}), \qquad \widehat{C}_{3}^{(n,M)} = \widetilde{C}_{3} + \bigO(n^{-10}), \qquad \widehat{C}_{4}^{(n,M)} = C_{4} + \bigO(n^{-10}),
\end{align*}
and the claim follows.
\end{proof}


\begin{proof}[End of the proof of Theorem \ref{thm:main thm}]
Let $M'>0$ be sufficiently large such that Lemmas \ref{lemma: S2km1} and \ref{lemma: asymp of S2k final} hold. Using \eqref{log Dn as a sum of sums} and Lemmas \ref{lemma: S0}, \ref{lemma: S2km1} and \ref{lemma: asymp of S2k final}, we conclude that for any $x_{1},\ldots,x_{p} \in \mathbb{R}$, there exists $\delta > 0$ such that
\begin{align*}
& \ln \mathcal{E}_{n} = S_{0}+S_{1}+S_{2}+S_{3} \\
& = M' \ln \Omega_{1} + (j_{-}-M'-1) \ln \Omega_{1} + \Big( br^{2b}n - j_{-} \Big) \ln  \Omega_{1}  + C_{2} \sqrt{n} + \widetilde{C}_{3} + \frac{1}{\sqrt{n}} C_{4} + \bigO(M^{4}n^{-1}) \\
& = \Big( br^{2b} \ln  \Omega_{1} \Big)n + C_{2} \sqrt{n} + C_{3} + \frac{1}{\sqrt{n}} C_{4} + \bigO(M^{4}n^{-1}),
\end{align*}
as $n \to +\infty$ uniformly for $u_{1} \in \{z \in \mathbb{C}: |z-x_{1}|\leq \delta\},\ldots,u_{p} \in \{z \in \mathbb{C}: |z-x_{p}|\leq \delta\}$, where $C_{3} = \widetilde{C}_{3}-\ln \Omega_{1}$. Using \eqref{function H1}--\eqref{function G2}, \eqref{def of omegaell} and \eqref{def of Omega j}, the constants $C_{2}$, $C_{3}$ and $C_{4}$ can be rewritten as in \eqref{asymp in main thm} after a change of variables.
\end{proof}
\section{Proof of Theorem \ref{thm:main thm edge}}\label{section:proof edge}

As in the proof of Theorem \ref{thm:main thm}, our starting point is formula \eqref{main exact formula}. 

\medskip Let 
\begin{align*}
& j_{-}:=\lceil \tfrac{n}{1+\epsilon} - \alpha \rceil,
\end{align*}
where $\epsilon > 0$ is a small constant independent of $n$. Using \eqref{main exact formula}, we write $\ln \mathcal{E}_{n}$ in $3$ parts
\begin{align}\label{log Dn as a sum of sums edge}
\ln \mathcal{E}_{n} = S_{0} + S_{1} + S_{2},
\end{align}
with 
\begin{align}
& S_{0} = \sum_{j=1}^{M'} \ln \bigg( 1+\sum_{\ell=1}^{m} \omega_{\ell} \frac{\gamma(\tfrac{j+\alpha}{b},nr_{\ell}^{2b})}{\Gamma(\tfrac{j+\alpha}{b})} \bigg), & & S_{1} = \sum_{j=M'+1}^{j_{-}-1} \hspace{-0.3cm} \ln \bigg( 1+\sum_{\ell=1}^{m} \omega_{\ell} \frac{\gamma(\tfrac{j+\alpha}{b},nr_{\ell}^{2b})}{\Gamma(\tfrac{j+\alpha}{b})} \bigg), \label{def of S0 and S1 edge} \\
& S_{2} = \sum_{j=j_{-}}^{n} \ln \bigg( 1+\sum_{\ell=1}^{m} \omega_{\ell} \frac{\gamma(\tfrac{j+\alpha}{b},nr_{\ell}^{2b})}{\Gamma(\tfrac{j+\alpha}{b})} \bigg), \label{def of S2 edge}
\end{align}
where $M'>0$ is a large integer independent of $n$. 
Recall the definition of $\Omega_{\ell}$ in \eqref{def of Omega j}.
\begin{lemma}\label{lemma: S0 edge}
For any $x_{1},\ldots,x_{m} \in \mathbb{R}$, there exists $\delta > 0$ such that
\begin{align}\label{asymp of S0 edge}
S_{0} = M' \ln \Omega_{1} + \bigO(e^{-cn}), \qquad \mbox{as } n \to + \infty,
\end{align}
uniformly for $u_{1} \in \{z \in \mathbb{C}: |z-x_{1}|\leq \delta\},\ldots,u_{m} \in \{z \in \mathbb{C}: |z-x_{m}|\leq \delta\}$.
\end{lemma}
\begin{proof}
The proof is identical to the proof of Lemma \ref{lemma: S0}.
\end{proof}

\begin{lemma}\label{lemma: S2km1 edge}
The constant $M'$ can be chosen sufficiently large such that the following holds. For any $x_{1},\ldots,x_{m} \in \mathbb{R}$, there exists $\delta > 0$ such that 
\begin{align*}
& S_{1} = (j_{-}-M'-1) \ln \Omega_{1} + \bigO(e^{-cn}),
\end{align*}
as $n \to +\infty$ uniformly for $u_{1} \in \{z \in \mathbb{C}: |z-x_{1}|\leq \delta\},\ldots,u_{m} \in \{z \in \mathbb{C}: |z-x_{m}|\leq \delta\}$.
\end{lemma}
\begin{proof}
The claim follows as in Lemma \ref{lemma: S2km1}.
\end{proof}
We now focus on $S_{2}$. For $j=1,\ldots,n$ and $k =1,\ldots,m$, define $a_{j}$, $z_{k}$, $\lambda_{j,k}$, $\lambda_{j}$ and $\eta_{j,k}$ as in \eqref{def etajl} with $r$ replaced by $b^{-\frac{1}{2b}}$, i.e. 
\begin{align*}
\lambda_{j,k} := \frac{z_{k}}{a_{j}} = \frac{bnr_{k}^{2b}}{j+\alpha}, \qquad \lambda_{j} := \frac{n}{j+\alpha}, \qquad \eta_{j,k} := (\lambda_{j,k}-1)\sqrt{\frac{2 (\lambda_{j,k}-1-\ln \lambda_{j,k})}{(\lambda_{j,k}-1)^{2}}},
\end{align*}
with $a_{j}:=\frac{j+\alpha}{b}$, $z_{k}:=nr_{k}^{2b}$ and $r_{1},\ldots,r_{m}$ as in the statement of Theorem \ref{thm:main thm edge}. Let $M:=M'\sqrt{\ln n}$. We split $S_{2}$ in two pieces as follows
\begin{align*}
& S_{2}=S_{2}^{(2)}+S_{2}^{(3)},
\end{align*}
where
\begin{align}\label{asymp prelim of S2kpvp edge}
& S_{2}^{(v)} = \sum_{j:\lambda_{j}\in I_{v}}  \ln \bigg( 1+\sum_{\ell=1}^{m} \omega_{\ell} \frac{\gamma(\tfrac{j+\alpha}{b},nr_{\ell}^{2b})}{\Gamma(\tfrac{j+\alpha}{b})} \bigg), \quad v=2,3,
\end{align}
and
\begin{align*}
I_{2} = [\tfrac{1}{1+\frac{\alpha}{n}},1+\tfrac{M}{\sqrt{n}}], \qquad I_{3} = (1+\tfrac{M}{\sqrt{n}},1+\epsilon].
\end{align*}
The sums $S_{2}^{(2)}$ and $S_{2}^{(3)}$ can also be rewritten using
\begin{align}\label{sums lambda j edge}
& \sum_{j:\lambda_{j}\in I_{3}} = \sum_{j=j_{-}}^{g_{-}-1}, \qquad \sum_{j:\lambda_{j}\in I_{2}} = \sum_{j= g_{-}}^{n}, 
\end{align}
where $g_{-} := \lceil \frac{n}{1+\frac{M}{\sqrt{n}}}-\alpha \rceil$. Let us also define
\begin{align*}
& \theta_{-}^{(n,M)} := g_{-} - \bigg( \frac{n}{1+\frac{M}{\sqrt{n}}} - \alpha \bigg) = \bigg\lceil \frac{n}{1+\frac{M}{\sqrt{n}}} - \alpha \bigg\rceil - \bigg( \frac{n}{1+\frac{M}{\sqrt{n}}} - \alpha \bigg).
\end{align*}
Clearly, $\theta_{-}^{(n,M)} \in [0,1)$.

\begin{lemma}\label{lemma:S2kp3p edge}
The constant $M'$ can be chosen sufficiently large such that the following holds. For any $x_{1},\ldots,x_{m} \in \mathbb{R}$, there exists $\delta > 0$ such that
\begin{align*}
S_{2}^{(3)} = & \; \Big( n - j_{-} - M\sqrt{n} + M^{2} -\alpha+\theta_{-}^{(n,M)} - M^{3}n^{-\frac{1}{2}} \Big) \ln  \Omega_{1} + \bigO(M^{4}n^{-1}),
\end{align*}
as $n \to +\infty$ uniformly for $u_{1} \in \{z \in \mathbb{C}: |z-x_{1}|\leq \delta\},\ldots,u_{m} \in \{z \in \mathbb{C}: |z-x_{m}|\leq \delta\}$.
\end{lemma}
\begin{proof}
The proof follows the proof of Lemma \ref{lemma:S2kp3p}, except that at the last step we need to use
\begin{align}
& \sum_{j = j_{-}}^{g_{-}-1} 1 = g_{-}-j_{-} = \bigg( \frac{n}{1+\frac{M}{\sqrt{n}}} - \alpha \bigg)+\theta_{-}^{(n,M)} -j_{-}  \nonumber \\
& = n - j_{-} - M\sqrt{n} + M^{2} -\alpha+\theta_{-}^{(n,M)} - M^{3} n^{-\frac{1}{2}} + \bigO(M^{4}n^{-1}), \quad \mbox{as } n \to + \infty. \label{sum jmin edge}
\end{align}
\end{proof}
For $k \in \{1,\ldots,m\}$ and $j \in \{j: \lambda_{j} \in I_{2}\}=\{g_{-},\ldots,n\}$, we define $M_{j,k} := \sqrt{n}(\lambda_{j,k}-1)$ and $M_{j} := \sqrt{n}(\lambda_{j}-1)$.

\begin{lemma}\label{lemma:Riemann sum edge}(Taken from \cite[Lemma 2.7]{Charlier 2d jumps})
Let $f \in C^{3}(\mathbb{R})$ be a function such that $|f|,|f'|,|f''|,|f'''|$ are bounded. As $n \to + \infty$ we have
\begin{align}
& \sum_{j=g_{-}}^{n}f(M_{j}) = \int_{0}^{M} f(t) dt \; \sqrt{n} - 2 \int_{0}^{M} tf(t) dt + \bigg( \frac{1}{2}-\theta_{-}^{(n,M)} \bigg)f(M) + \bigg( \frac{1}{2} + \alpha \bigg)f(0) \nonumber \\
& + \frac{1}{\sqrt{n}}\bigg[ 3 \int_{0}^{M}t^{2}f(t)dt + \bigg( \frac{1}{12}+\frac{\theta_{-}^{(n,M)}(\theta_{-}^{(n,M)}-1)}{2} \bigg)f'(M) - \frac{1+6\alpha+6\alpha^{2}}{12}f'(0) \bigg] + \bigO(M^{4}n^{-1}). \label{sum f asymp 2 edge}
\end{align}
\end{lemma} 

\begin{lemma}\label{lemma:S2kp2p edge}
For any $x_{1},\ldots,x_{p} \in \mathbb{R}$, there exists $\delta > 0$ such that
\begin{align*}
&  S_{2}^{(2)} = \widetilde{C}_{2}^{(M)}\sqrt{n} + \widetilde{C}_{3}^{(n,M)} + \frac{1}{\sqrt{n}}\widetilde{C}_{4}^{(n,M)} + \bigO(M^{4}n^{-1}), \\
& \widetilde{C}_{2}^{(M)} = \int_{0}^{M} f_{1}(t) dt,  \\
& \widetilde{C}_{3}^{(n,M)} = \int_{0}^{M} \Big( -2t f_{1}(t)+f_{2}(t) \Big) dt + \bigg( \frac{1}{2}-\theta_{-}^{(n,M)} \bigg)f_{1}(M)+\bigg( \frac{1}{2} + \alpha \bigg)f_{1}(0), \\
& \widetilde{C}_{4}^{(n,M)} = \int_{0}^{M}\Big( 3t^{2}f_{1}(t)-2tf_{2}(t)+f_{3}(t) \Big) dt + \bigg( \frac{1}{2}-\theta_{-}^{(n,M)} \bigg)f_{2}(M)+\bigg( \frac{1}{2}+\alpha \bigg)f_{2}(0) \\
& + \bigg( \frac{1}{12}+\frac{\theta_{-}^{(n,M)}(\theta_{-}^{(n,M)}-1)}{2} \bigg) f_{1}'(M) - \frac{1+6\alpha+6\alpha^{2}}{12} f_{1}'(0) ,
\end{align*}
as $n \to +\infty$ uniformly for $u_{1} \in \{z \in \mathbb{C}: |z-x_{1}|\leq \delta\},\ldots,u_{p} \in \{z \in \mathbb{C}: |z-x_{p}|\leq \delta\}$, where $g$, $f_{1}$, $f_{2}$ and $f_{3}$ are as in the statement of Lemma \ref{lemma:S2kp2p} with $r$ replaced by $b^{-\frac{1}{2b}}$. 
\end{lemma}
\begin{proof}
The first part of the proof is identical to the beginning of the proof of Lemma \ref{lemma:S2kp2p}, except that one needs to replace $r$ and $g_{+}$ by $b^{-\frac{1}{2b}}$ and $n$, respectively. In particular, we find
\begin{align}
S_{2}^{(2)} & = \Sigma_{1}^{(n)}+\frac{1}{\sqrt{n}}\Sigma_{2}^{(n)}+\frac{1}{n}\Sigma_{3}^{(n)} + \bigO(n^{-1}), \qquad \mbox{as } n \to + \infty, \label{asymp of S2kp2p in proof edge}
\end{align}
where $\Sigma_{1}^{(n)}$, $\Sigma_{2}^{(n)}$ and $\Sigma_{3}^{(n)}$ are in the proof of Lemma \ref{lemma:S2kp2p} with $r$ and $g_{+}$ replaced by $b^{-\frac{1}{2b}}$ and $n$, respectively. The asymptotics of these sums can then be obtained using Lemma \ref{lemma:Riemann sum edge}. After a computation, we then find the claim.
\end{proof}
\begin{lemma}\label{lemma: asymp of S2k final edge}
The constant $M'$ can be chosen sufficiently large such that the following holds. For any $x_{1},\ldots,x_{p} \in \mathbb{R}$, there exists $\delta > 0$ such that
\begin{align*}
& S_{2} =  ( n - j_{-} ) \ln  \Omega_{1}  + C_{2} \sqrt{n} + \widetilde{C}_{3} + \frac{1}{\sqrt{n}} C_{4} + \bigO(M^{4}n^{-1}),
\end{align*}
as $n \to +\infty$ uniformly for $u_{1} \in \{z \in \mathbb{C}: |z-x_{1}|\leq \delta\},\ldots,u_{p} \in \{z \in \mathbb{C}: |z-x_{p}|\leq \delta\}$, where
\begin{align*}
& C_{2} = \int_{0}^{+\infty} \big( f_{1}(t)-\ln \Omega_{1} \big) dt, \\
& \widetilde{C}_{3} = \bigg( \frac{1}{2} - \alpha \bigg) \ln \Omega_{1} + \int_{0}^{+\infty} (-2t[f_{1}(t)-\ln \Omega_{1}]+f_{2}(t)) dt + \bigg( \frac{1}{2}+\alpha \bigg)f_{1}(0), \\
& C_{4} = \int_{0}^{+\infty} \big( 3t^{2}(f_{1}(t)-\ln(\Omega_{1}))-2tf_{2}(t)+f_{3}(t) \big)dt - \frac{1+6\alpha+6\alpha^{2}}{12}f_{1}'(0) + \bigg( \frac{1}{2}+\alpha \bigg)f_{2}(0).
\end{align*}
\end{lemma}
\begin{proof}
By combining Lemmas \ref{lemma:S2kp3p edge} and \ref{lemma:S2kp2p edge}, we have
\begin{align*}
& S_{2} = ( n - j_{-} ) \ln  \Omega_{1} + \widehat{C}_{2}^{(M)} \sqrt{n} + \widehat{C}_{3}^{(n,M)} + \frac{1}{\sqrt{n}} \widehat{C}_{4}^{(n,M)} + \bigO(M^{4}n^{-1}),
\end{align*}
as $n \to +\infty$ uniformly for $u_{1} \in \{z \in \mathbb{C}: |z-x_{1}|\leq \delta\},\ldots,u_{p} \in \{z \in \mathbb{C}: |z-x_{p}|\leq \delta\}$, where
\begin{align*}
& \widehat{C}_{2}^{(M)} := \widetilde{C}_{2}^{(M)} - M \ln \Omega_{1}, \\
& \widehat{C}_{3}^{(n,M)} := \widetilde{C}_{3}^{(n,M)} + \Big( M^{2} -\alpha+\theta_{-}^{(n,M)}  \Big) \ln \Omega_{1}, \\
& \widehat{C}_{4}^{(n,M)} := \widetilde{C}_{4}^{(n,M)} - M^{3} \ln \Omega_{1}.
\end{align*}
Provided $M'$ is chosen sufficiently large, as $n \to + \infty$ we get
\begin{align*}
\widehat{C}_{2}^{(M)} = C_{2} + \bigO(n^{-10}), \qquad \widehat{C}_{3}^{(n,M)} = \widetilde{C}_{3} + \bigO(n^{-10}), \qquad \widehat{C}_{4}^{(n,M)} = C_{4} + \bigO(n^{-10}),
\end{align*}
and the claim follows.
\end{proof}

\begin{proof}[End of the proof of Theorem \ref{thm:main thm edge}]
Let $M' > 0$ be sufficiently large such that Lemmas \ref{lemma: S2km1 edge} and \ref{lemma: asymp of S2k final edge} hold. Using \eqref{log Dn as a sum of sums edge} and Lemmas \ref{lemma: S0 edge}, \ref{lemma: S2km1 edge} and \ref{lemma: asymp of S2k final edge}, we conclude that for any $x_{1},\ldots,x_{p} \in \mathbb{R}$, there exists $\delta > 0$ such that
\begin{align*}
& \ln \mathcal{E}_{n} = S_{0}+S_{1}+S_{2} \\
& = M' \ln \Omega_{1} + (j_{-}-M'-1) \ln \Omega_{1} + ( n - j_{-} ) \ln  \Omega_{1}  + C_{2} \sqrt{n} + \widetilde{C}_{3} + \frac{1}{\sqrt{n}} C_{4} + \bigO(M^{4}n^{-1}) \\
& = n \ln  \Omega_{1} + C_{2} \sqrt{n} + C_{3} + \frac{1}{\sqrt{n}} C_{4} + \bigO(M^{4}n^{-1}),
\end{align*}
as $n \to +\infty$ uniformly for $u_{1} \in \{z \in \mathbb{C}: |z-x_{1}|\leq \delta\},\ldots,u_{p} \in \{z \in \mathbb{C}: |z-x_{p}|\leq \delta\}$, where $C_{3} = \widetilde{C}_{3}-\ln \Omega_{1}$. Using \eqref{function H1}--\eqref{function G2}, \eqref{def of omegaell} and \eqref{def of Omega j}, the constants $C_{2}$, $C_{3}$ and $C_{4}$ can be rewritten as in \eqref{asymp in main thm edge} after a simple change of variables. This concludes the proof of Theorem \ref{thm:main thm edge}.
\end{proof}

\appendix
\section{Uniform asymptotics of the incomplete gamma function}\label{section:uniform asymp gamma}
\begin{lemma}\label{lemma:various regime of gamma}(From \cite[formula 8.11.2]{NIST}).
Let $a>0$ be fixed. As $z \to +\infty$,
\begin{align*}
\gamma(a,z) = \Gamma(a) + \bigO(e^{-\frac{z}{2}}).
\end{align*}
\end{lemma}

\begin{lemma}\label{lemma: uniform}(From \cite[Section 11.2.4]{Temme}).
We have
\begin{align*}
& \frac{\gamma(a,z)}{\Gamma(a)} = \frac{1}{2}\mathrm{erfc}(-\eta \sqrt{a/2}) - R_{a}(\eta), \qquad R_{a}(\eta) = \frac{e^{-\frac{1}{2}a \eta^{2}}}{2\pi i}\int_{-\infty}^{\infty}e^{-\frac{1}{2}a u^{2}}g(u)du,
\end{align*}
where $\mathrm{erfc}$ is defined in \eqref{def of erfc}, 
\begin{align}\label{lol8}
& \lambda = \frac{z}{a}, \qquad \eta = (\lambda-1)\sqrt{\frac{2 (\lambda-1-\ln \lambda)}{(\lambda-1)^{2}}}, \qquad g(u) := \frac{dt}{du} \frac{1}{\lambda -t} + \frac{1}{u +i\eta},
\end{align}
with $t$ and $u$ being related by the bijection $t \mapsto u$ from $\mathcal{L} := \{\frac{\theta}{\sin \theta} e^{i\theta} : - \pi < \theta < \pi\}$ to $\R$ given by
$$u = -i(t-1) \sqrt{\frac{2(t-1-\ln t)}{(t-1)^2}}, \qquad  t \in \mathcal{L},$$
and the principal branch is used for the roots. Furthermore, as $a \to + \infty$, uniformly for $z \in [0,\infty)$,
\begin{align}\label{asymp of Ra}
& R_{a}(\eta) \sim \frac{e^{-\frac{1}{2}a \eta^{2}}}{\sqrt{2\pi a}}\sum_{j=0}^{\infty} \frac{c_{j}(\eta)}{a^{j}}, 
\end{align}
where all coefficients $c_{j}(\eta)$ are bounded functions of $\eta \in \mathbb{R}$ (i.e. bounded for $\lambda \in [0,+\infty)$). The first two coefficients are given by (see \cite[p. 312]{Temme})
\begin{align*}
c_{0}(\eta) = \frac{1}{\lambda-1}-\frac{1}{\eta}, \qquad c_{1}(\eta) = \frac{1}{\eta^{3}}-\frac{1}{(\lambda-1)^{3}}-\frac{1}{(\lambda-1)^{2}}-\frac{1}{12(\lambda-1)}.
\end{align*}
In particular, the following hold:
\item[(i)] Let $z=\lambda a$ and let $\delta >0$ be fixed. As $a \to +\infty$, uniformly for $\lambda \geq 1+\delta$,
\begin{align*}
\gamma(a,z) = \Gamma(a)\big(1 + \bigO(e^{-\frac{a \eta^{2}}{2}})\big).
\end{align*}
\item[(ii)] Let $z=\lambda a$. As $a \to +\infty$, uniformly for $\lambda$ in compact subsets of $(0,1)$,
\begin{align*}
\gamma(a,z) = \Gamma(a)\bigO(e^{-\frac{a \eta^{2}}{2}}).
\end{align*}
\end{lemma} 

\paragraph{Acknowledgements.} CC acknowledges support from the Ruth and Nils-Erik Stenb\"ack Foundation, the Novo Nordisk Fonden Project, Grant 0064428, and the Swedish Research Council, Grant No. 2021-04626. JL acknowledges support from the European Research Council, Grant Agreement No. 682537, the Swedish Research Council, Grant No. 2021-03877, and the Ruth and Nils-Erik Stenb\"ack Foundation.

\end{document}